\documentclass[runningheads,a4paper,english]{llncs}

\usepackage[margin=4.0cm]{geometry}
\usepackage{lineno,hyperref}
\usepackage{color}
\usepackage{float}
\usepackage{bm}
\usepackage{amsmath}
\usepackage{amssymb}
\usepackage{graphicx}
\usepackage{setspace}
\usepackage{subscript}
\usepackage{upgreek}
\usepackage{babel}
\usepackage[numbers,square]{natbib}
\usepackage[algo2e,ruled]{algorithm2e} 
\usepackage[font=scriptsize,skip=0pt]{caption}

\newcommand{\diag}{\mathop{\mathrm{diag}}}

\SetKwInOut{Input}{Input}
\SetKwInOut{Output}{Output}
\SetKwInOut{Initialize}{Initialize}
\DeclareMathAlphabet{\mathpzc}{OT1}{pzc}{m}{it}
\usepackage{authblk}

\newcommand{\RN}[1]{%
  \textup{\uppercase\expandafter{\romannumeral#1}}%
}
\newcommand{\keywords}[1]{\par\addvspace\baselineskip
\noindent\keywordname\enspace\ignorespaces#1}
\authorrunning{E. Bradford et al.}
\titlerunning{Hybrid GP modelling applied to economic SMPC}

\begin{document}

\mainmatter 

\title{Hybrid Gaussian Process Modeling Applied to Economic Stochastic Model Predictive Control of Batch Processes}

\author{Eric Bradford$^*$, Lars Imsland$^*$, Marcus Reble$^\dagger$, and \\ Ehecatl Antonio del Rio-Chanona$^\ddagger$}

\institute{$^*$ Engineering Cybernetics, Norwegian University of Science and Technology, Trondheim, Norway  \\
$^\dagger$ BASF SE, 67056 Ludwigshafen, Germany  \\
$^\ddagger$ Centre for Process Systems Engineering (CPSE), Department of Chemical Engineering, Imperial College London, UK (a.del-rio-chanona@imperial.ac.uk)}

\maketitle

\begin{abstract}
Nonlinear model predictive control (NMPC) is an efficient approach for the control of nonlinear multivariable dynamic systems with constraints, which however requires an accurate plant model. Plant models can often be determined from first principles, parts of the model are however difficult to derive using physical laws alone. In this paper a hybrid Gaussian process (GP) first principles modeling scheme is proposed to overcome this issue, which exploits GPs to model the parts of the dynamic system that are difficult to describe using first principles. GPs not only give accurate predictions, but also quantify the residual uncertainty of this model. It is vital to account for this uncertainty in the control algorithm, to prevent constraint violations and performance deterioration. Monte Carlo samples of the GPs are generated offline to tighten constraints of the NMPC to ensure joint probabilistic constraint satisfaction online. Advantages of our method include fast online evaluation times, possibility to account for online learning alleviating conservativeness, and  exploiting the flexibility of GPs and the data efficiency of first principle models. The algorithm is verified on a case study involving a challenging semi-batch bioreactor.  
\keywords{Uncertain dynamic systems, Back-offs, Model-based nonlinear control, Machine learning, Chance constraints, Robust control \\
This is a pre-peer-review, pre-copyedit version of an article published
in Recent Advances in Model Predictive Control, Lecture Notes in Control and Information Sciences, vol 485, pp. 191-218 (2021). The final authenticated 
version is available online at: \url{http://dx.doi.org/	10.1007/978-3-030-63281-6_8}.}
\end{abstract}

\section{Introduction}
Model predictive control (MPC) refers to a class of control methods, which makes explicit use of a process model to determine a sequence of control actions to take at each sampling time. Feedback is achieved through the repeated update of the initial state. MPC is especially useful to deal with multivariable control problems and important process constraints \cite{Maciejowski2002}. Many processes are highly nonlinear and may be operated at unsteady state, which motivates the use of nonlinear MPC (NMPC). In particular NMPC applications based on first principles models are becoming increasingly popular due to the advent of improved optimization methods and the availability of more models \cite{Biegler2010}. In this paper we focus on \textit{finite horizon} control problems, for which chemical batch processes are a particularly important example. These are employed in many different chemical sectors due to their inherent flexibility. Previous works for batch processes include NMPC based on the extended and unscented Kalman filter \cite{Nagy2007a,Bradford2018c}, polynomial chaos expansions \cite{Mesbah2014,Bradford2019d}, and multi-stage NMPC \cite{Lucia2013}.

A major limitation of NMPC in practice is the requirement of an accurate dynamic plant model, which has been cited to take up to 80$\%$ of the MPC commissioning effort \cite{Sun2013}. The required dynamic model for NMPC is often derived from first principles taking advantage of the available prior knowledge of the process \cite{Nagy2007b}. While this can be an efficient modeling approach, often parts of the model are notoriously difficult to represent using physical laws. In addition, modelling certain phenomena may require excessive amounts of computational time. For example in chemical engineering hybrid models have been developed to capture chemical reaction kinetics \cite{Teixeira2007,Psichogios1992}, the complex mechanics of catalyst deactivation \cite{Azarpour2017}, or for the correction of first principles models using available measurements \cite{Hermanto2011,Bhutani2006}. Most hybrid modelling applications have been focused on using neural networks (NNs). In this paper we propose to use Gaussian processes (GPs) instead \cite{Rasmussen2005} due to their ability to not only provide accurate predictions, but also provide a measure of uncertainty for these predictions difficult to obtain by other nonlinear modeling approaches \cite{Kocijan2005}. It is important to account for this measure of uncertainty to avoid constraint violations and performance deterioration. To consider uncertainty for NMPC formulations explicitly robust MPC \cite{Campo1987} and stochastic MPC \cite{Farina2016} approaches have been developed. Previous works on using GPs for hybrid modelling have been mainly focused on linear ordinary differential equation systems of first- and second order that can be solved exactly, see for example \cite{Sarkka2018,Alvarez2009,Lawrence2007}.           

GP-based MPC was first proposed in \cite{Murray-Smith2003}, in which the GP is recursively updated for reference tracking. In \cite{Kocijan2004,Kocijan2005a} it is proposed to identify the GP offline and apply it online for NMPC instead. The variance therein is constrained to avoid the NMPC steering into regions of high uncertainty. Furthermore, GPs have been used to overcome deviations between the approximate plant model utilized and the real plant model \cite{Klenske2016,Maciejowski2013}. GPs may also act as an efficient surrogate to estimate the mean and variance required for stochastic NMPC \cite{Bradford2018b}. Applications of GP-based MPC includes the control of an unmanned quadrotor \cite{Cao2017}, the control of a gas-liquid separation process \cite{Likar2007}, and the steering of miniature cars \cite{Hewing2018}. While these works show the feasibility of GP-based MPC, most formulations use stochastic uncertainty propagation to account for the uncertainty measure provided by the GP, e.g. \cite{Kocijan2004,Kocijan2005a,Hewing2018,Cao2017}. An overview of these approaches can be found in \cite{Hewing2017}. Major limitations of stochastic propagation is open-loop uncertainty growth, no known method for exact propagation of stochastic uncertainties, and significantly increased computation times. Recently, several papers have proposed alternative techniques to consider the GP uncertainty measure. \cite{Koller2018a} propagate ellipsoidal sets using linearization and accounting for the linearization error by employing Lipschitz constant, which is however relatively conservative. \cite{Maiworm2018} use a robust MPC approach by bounding the one-step ahead error from the GP, while \cite{Soloperto2018} suggest a robust control approach for linear systems to account for unmodelled nonlinearities. This approach may however be infeasible if the deviation between the nonlinear system and linear system is too large. 

In this paper we extend a method first proposed in \cite{Bradford2019b,Bradford2019e} to the hybrid modelling case. The approach determines explicit back-offs to tighten constraints offline using closed-loop Monte Carlo (MC) simulations for \textit{finite horizon} control problems. These then in turn guarantee the satisfaction of probabilistic constraints online. There are several advantages of this approach including avoidance of closed-loop uncertainty growth, fast online computational times, probabilistic guarantees on constraint satisfaction, and explicit consideration of online learning to alleviate conservativeness. In addition, sampled GPs lead to deterministic models that can be easily handled in a hybrid modelling framework. In contrast, obtaining statistical moments for stochastic uncertainty propagation for hybrid models is difficult.       

The paper is comprised of the following sections. In Section \ref{sec:prob_def} the problem definition is given. Thereafter, in Section \ref{sec:solution_approach} we outline the solution approach. Section \ref{sec:case_study} outlines the semi-batch bioprocess case study to be solved, while in Section \ref{sec:results} results and discussions for this case study are presented. Section \ref{sec:conclusions} concludes the paper. 

\section{Problem definition} \label{sec:prob_def}
The dynamic system in this paper is assumed to be given by a discrete-time nonlinear equation system with additive disturbance noise and an unknown function $\mathbf{q}(\cdot,\mathbf{u}_k)$:
\begin{equation} \label{eq:f_x}
    \mathbf{x}_{k+1} = \mathbf{F}(\mathbf{x}_k,\mathbf{u}_k,\mathbf{q}(\cdot,\mathbf{u}_k)) + \bm{\upomega}_k, \quad \mathbf{x}(0) \sim \mathcal{N}(\bm{\upmu}_{\mathbf{x}_0},\bm{\Sigma}_{\mathbf{x}_0}) 
\end{equation}
where $\mathbf{x}_k \in \mathbb{R}^{n_{\mathbf{x}}}$ represent the states, $\mathbf{u}_k$ denotes the control inputs, $\mathbf{q}:\mathbb{R}^{n_{\mathbf{x}}} \times \mathbb{R}^{n_{\mathbf{u}}} \rightarrow \mathbb{R}^{n_{\mathbf{q}}}$ are unknown nonlinear functions, and $\mathbf{F}:\mathbb{R}^{n_{\mathbf{x}}} \times \mathbb{R}^{n_{\mathbf{u}}} \times \mathbb{R}^{n_{\mathbf{q}}} \rightarrow \mathbb{R}^{n_{\mathbf{x}}}$ are known nonlinear functions. The initial condition $\mathbf{x}(0)$ is assumed to follow a Gaussian distribution with mean $\bm{\upmu}_{\mathbf{x}_0}$ and covariance $\bm{\Sigma}_{\mathbf{x}_0}$. Additive disturbance noise is denoted by $\bm{\upomega}_k$, which is assumed to follow a Gaussian distribution with zero mean and covariance matrix $\bm{\Sigma}_{\bm{\upomega}}$, $\bm{\upomega}_k \sim \mathcal{N}(\mathbf{0},\bm{\Sigma}_{\bm{\upomega}})$.   

Note most first principles models are given in continuous-time, which has important implications on the unknown function $\mathbf{q}(\cdot,\mathbf{u}_k)$. For example, this model needs to be well-identified not only at these discrete times. Let $\delta_{t}=t_k - t_{k-1}$ be a constant sampling time at which measurements are taken. The corresponding continuous-time model to $\mathbf{F}(\cdot)$ is represented by $\mathbf{f}(\cdot)$:      
\begin{equation} \label{eq:disc_x}
    \mathbf{F}(\mathbf{x}_k,\mathbf{u}_k,\mathbf{q}(\cdot,\mathbf{u}_k)) = \int_{t_{k}}^{t_{k+1}} \mathbf{f}(\mathbf{x}(t),\mathbf{u}_k,\mathbf{q}(\mathbf{x}(t),\mathbf{u}_k)) \text{d}t + \mathbf{x}_k
\end{equation}
where $\mathbf{x}_k = \mathbf{x}(t_k)$ is the value of the state at discrete-time $k$.

In general $\mathbf{q}(\cdot)$ may be composed of $n_{\mathbf{q}}$ separate scalar functions, such that \\ $\mathbf{q}(\mathbf{x},\mathbf{u})=[q_1(\mathbf{q}^{in}_1(\mathbf{x},\mathbf{u})),\ldots,q_{n_{\mathbf{q}}}(\mathbf{q}^{in}_{n_{\mathbf{q}}}(\mathbf{x},\mathbf{u}))]^{\sf T}$ with $n_q$ separate input functions $\mathbf{q}^{in}_{i}:\mathbb{R}^{n_{\mathbf{x}}} \times \mathbb{R}^{n_{\mathbf{u}}} \rightarrow \mathbb{R}^{n_{\mathbf{q}^{in}_{i}}}$ for $i=1,\ldots,n_{\mathbf{q}}$. Note these input functions are assumed to be known, since commonly the unknown function denotes an unmodelled physical process, for which the inputs are known a priori. The input dimension $n_{\mathbf{q}^{in}_{i}}$ is usually much lower than the dimension of states and control inputs combined, and therefore modelling these components can be considerably more data efficient than determining the full state space model from data instead.    

The variable $\bm{\upomega}_k$ represents additive disturbance noise with zero mean and a covariance matrix $\bm{\Sigma}_{\bm{\upomega}}$. The measurement at discrete time $t=t_k$ can be expressed as follows:
\begin{equation} \label{eq:y}
    \mathbf{y}_k = \mathbf{H} \mathbf{x}_k + \bm{\upnu}_k
\end{equation}
where $\mathbf{y}_k$ is the corresponding measurement, $\mathbf{H} \in \mathbb{R}^{n_{\mathbf{y}} \times n_\mathbf{x}}$ is the linear observation model, and $\bm{\upnu}_k$ denotes additive measurement noise with zero mean and a covariance matrix $\bm{\Sigma}_{\bm{\upnu}}$.  

The aim of the control problem the minimization of a finite-horizon cost function:
\begin{equation} \label{eq:objdef}
    V_T(\mathbf{x}_0,\mathbf{U}) = \mathbb{E}\left[\sum_{k=0}^{T-1} \ell(\mathbf{x}_k,\mathbf{u}_k) + \ell_f(\mathbf{x}_T) \right]
\end{equation}
where $T \in \mathbb{N}$ is the time horizon, $\mathbf{U}=[\mathbf{u}_0,\ldots,\mathbf{u}_{T-1}]^{\sf T} \in \mathbb{R}^{T \times n_{\mathbf{u}}}$ is a joint matrix over all control inputs for time horizon $T$, $\ell:\mathbb{R}^{n_{\mathbf{x}}} \times \mathbb{R}^{n_{\mathbf{u}}} \rightarrow \mathbb{R}$ represent the stage costs, and $\ell_f:\mathbb{R}^{n_{\mathbf{x}}}$ is the terminal cost. 

The control inputs are subject to hard constraints:
\begin{equation} \label{eq:ucon}
    \mathbf{u}_k \in \mathbb{U} \quad \forall k \in \{0,\ldots,T-1\} 
\end{equation}

The states are subject to the satisfaction of a joint nonlinear chance constraint over the time horizon $T$, which can be stated as:
\begin{subequations} \label{eq:xcon}
\begin{align}
    & \mathbb{P} \left\{ \bigcap^T_{k=0} \{ \mathbf{x}_k \in \mathbb{X}_k \}    \right\} \geq 1-\epsilon
\end{align}
where $\mathbb{X}_t$ is defined as:
\begin{align}
    & \mathbb{X}_k = \{ \mathbf{x} \in \mathbb{R}^{n_{\mathbf{x}}} \mid g_j^{(k)}(\mathbf{x}) \leq 0, j=1,\ldots,n_g \}
\end{align}
\end{subequations}
The state constraint requires the joint event of all $\mathbf{x}_k$ for all $k \in \{0,\ldots,T\}$ fulfilling the nonlinear constraint sets $\mathbb{X}_k$ to have a probability greater than $1-\epsilon$.  

It is assumed that $\mathbf{f}(\cdot)$ and $\mathbf{q}^{in}_{i}$ for $i=1,\ldots,n_{\mathbf{q}}$ are known, while $\mathbf{q}(\cdot)$ is unknown and needs to be identified from data. We assume we are given $N$ noisy measurements according to Equation \ref{eq:y}, which is given by the following two matrices:
\begin{subequations} \label{eq:datasetdef}
\begin{align}
    & \mathbf{Z} = [\mathbf{z}^{(1)}_k,\ldots,\mathbf{z}^{(N)}_k]^{\sf T} \in \mathbb{R}^{N \times n_{\mathbf{z}}} \\
    & \mathbf{Y} = [\mathbf{y}^{(1)}_{k+1},\ldots,\mathbf{y}^{(N)}_{k+1}]^{\sf T} \in \mathbb{R}^{N \times n_{\mathbf{y}}}
\end{align}
\end{subequations}
where $\mathbf{z}^{(i)}_k=(\mathbf{x}^{(i)}_k,\mathbf{u}^{(i)}_k)$ is a tuple of $\mathbf{x}^{(i)}_k$ and $\mathbf{u}^{(i)}_k$, which are the $i$-th input of the data at discrete time $k$ with corresponding noisy measurements given by $\mathbf{y}_{k+1}^{(i)}$ at discrete time $k+1$. The matrix $\mathbf{Z}$ is a collection of input data with the corresponding noisy observations collected in $\mathbf{Y}$.    

The noise in this problem arises in part from the additive disturbance noise $\bm{\upomega}$ and from the noisy initial condition $\mathbf{x}_0$. The more important source of noise however originates from the unknown function $\mathbf{q}(\cdot)$, which is identified from only \textit{finite} amount of data. To solve this problem we train GPs to approximate $\mathbf{q}(\cdot)$ from the data in Equation \ref{eq:datasetdef}. In the next section we first introduce GPs to model the function $\mathbf{q}(\cdot)$, which then also represent the residual uncertainty of $\mathbf{q}(\cdot)$. This uncertainty representation is thereafter exploited to obtain the required stochastic constraint satisfaction of the closed-loop system.      

\section{Solution approach} \label{sec:solution_approach}
\subsection{Gaussian process hybrid model training} \label{sec:hybrid_GP_training}
In this section we introduce GPs to obtain a probabilistic model description for $\mathbf{q}(\cdot)$. For more information on GPs refer to \cite{Rasmussen2005}. For this we use a separate GP for each component $q_i(\mathbf{q}^{in}_i(\mathbf{x},\mathbf{u}))$ for $i=1,\ldots,n_{\mathbf{q}}$, which is standard practice in the GP community to handle multivariate outputs \cite{Deisenroth2011}. Let $i$ refer to the GP of function $i$ of $\mathbf{q}$. 

A GP describes a distribution over functions and can be viewed as a generalization of multivariate Gaussian distributions. We assume $q_i(\cdot)$ is distributed as a GP with mean function $m_i(\cdot)$ and covariance function $k_i(\cdot,\cdot)$, which fully specifies the GP prior: 
\begin{equation}
    q_i(\cdot) \sim \mathcal{GP}(m_i(\cdot),k_i(\cdot,\cdot)) \label{eq:fGPdist}
\end{equation}

The \textit{choice} of the mean and covariance function define the GP prior. In this study we use a zero mean function and the squared-exponential (SE) covariance function:
\begin{subequations}
\begin{align}
    & m_i(\mathbf{q}_i^{in}) := 0 \label{eq:meanfundef} \\
    & k_i(\mathbf{q}_i^{in},\mathbf{q}_i^{'in}) := \zeta^2_i \exp \left( -\frac{1}{2}(\mathbf{z} - \mathbf{z}')^{\sf T}\bm{\Lambda}^{-2}_i(\mathbf{z} - \mathbf{z}') \right) \label{eq:covfundef}
\end{align}
\end{subequations}
where $\mathbf{q}_i^{in},\mathbf{q}_i^{'in} \in \mathbb{R}^{n_{\mathbf{z}}}$ are arbitrary inputs, $\zeta^2_i$ denotes the covariance magnitude, and $\bm{\Lambda}_i^{-2} := \diag(\lambda_1^{-2},\ldots,\lambda_{n_{\mathbf{z}}}^{-2})$ is a scaling matrix.

\begin{remark}[Prior assumptions]
Zero mean can be realized by normalizing the data. Choosing the SE covariance function assumes the function to be modelled $q_i(\cdot)$ to be smooth and stationary.
\end{remark}

Now assume we are given $N$ values of $q_i(\cdot)$ , which we jointly denote as $\mathbf{Q}_i=[q_i^{(1)},\ldots,q_i^{(N)}]^{\sf T} \in \mathbb{R}^N$ and assume these correspond to their $q(\cdot)$ values at the inputs defined in $\mathbf{Z}$ in Equation \ref{eq:datasetdef}. The corresponding input response matrices to $\mathbf{Z}$ are then given by $\mathbf{Q}^{in}_i=[\mathbf{q}_i^{in}(\mathbf{z}^{(1)}_k),\ldots,\mathbf{q}_i^{in}(\mathbf{z}^{(N)}_k)]^{\sf T} \in \mathbb{R}^{N \times n_{\mathbf{q}_i^{in}}}$ . According to the GP prior the data vectors $\mathbf{Q}_i$ follow the following multivariate normal distribution:
\begin{equation} \label{eq:prior_Qi}
    \mathbf{Q}_i \sim \mathcal{N}(\mathbf{0},\bm{\Sigma}_{\mathbf{Q}_i})
\end{equation}
where $[\bm{\Sigma}_{\mathbf{Q}_i}]_{lm} = k_i(\mathbf{q}_i^{in}(\mathbf{z}^{(l)}_k),\mathbf{q}_i^{in}(\mathbf{z}^{(m)}_k)) + \sigma^2_{\nu i} \delta_{l,m}$ for each pair $(l,m)\in\{ 1,\ldots,N\}^2$. In essence this places a likelihood on the training dataset based on the continuity and smoothness assumptions made by the choice of the covariance function. The characteristic length-scales and hyperparameters introduced are jointly denoted by $\bm{\Psi}_i=[\lambda_1^{2},\ldots,\lambda_{n_{\mathbf{q}_i^{in}}}^{2},\zeta^2_i,\sigma^2_{\nu i}]^{\sf T}$. 

\textit{Given} a value of $\mathbf{Q}_i$ we can further determine a likelihood for values not part of $\mathbf{Q}_i$ using conditioning. Let $\hat{\mathbf{Q}_i}$ represent $\hat{N}$ such values at the inputs $\hat{\mathbf{Q}}^{in}_i=[\mathbf{q}_i^{in}(\hat{\mathbf{z}}^{(1)}_k),\ldots,\mathbf{q}_i^{in}(\hat{\mathbf{z}}^{(\hat{N})}_k)]^{\sf T} \in \mathbb{R}^{\hat{N} \times n_{\mathbf{q}_i^{in}}}$. From the prior GP assumption $\mathbf{Q}_i$ and $\hat{\mathbf{Q}}_i$ follow a joint Gaussian distribution:
\begin{align}
    & \begin{bmatrix}
    \mathbf{Q}_i   \\
    \hat{\mathbf{Q}}_i 
\end{bmatrix} \sim \mathcal{N}\left(\begin{bmatrix}
    \mathbf{0}   \\
    \mathbf{0} 
\end{bmatrix},\begin{bmatrix}
    \bm{\Sigma}_{\mathbf{Q}_i} & \bm{\Sigma}_{\hat{\mathbf{Q}}_i,\mathbf{Q}_i}^{\sf T}  \\
    \bm{\Sigma}_{\hat{\mathbf{Q}}_i,\mathbf{Q}_i} & \bm{\Sigma}_{\hat{\mathbf{Q}}_i}
\end{bmatrix}\right) \label{eq:jointnormal}
\end{align}    
where $[\bm{\Sigma}_{\hat{\mathbf{Q}}_i}]_{lm} = k(\mathbf{q}_i^{in}(\hat{\mathbf{z}}^{(l)}_k),\mathbf{q}_i^{in}(\hat{\mathbf{z}}^{(m)}_k)) + \sigma^2_{\nu i} \delta_{l,m}$ for each pair $(l,m)\in\{ 1,\ldots,\hat{N}\}^2$ and $[\bm{\Sigma}_{\hat{\mathbf{Q}}_i,\mathbf{Q}_i}]_{lm} = k(\mathbf{q}_i^{in}(\hat{\mathbf{z}}^{(l)}_k),\mathbf{q}_i^{in}(\mathbf{z}^{(m)}_k))$ for each pair $(l,m)\in \{ 1,\ldots,N\} \times \{ 1,\ldots,\hat{N}\}$ \cite{Rasmussen2005}. 

The likelihood of $\hat{\mathbf{Q}}_i$ conditioning on $\mathbf{Q}_i$ is then given by:
\begin{align}
    \hat{\mathbf{Q}}_i \sim \mathcal{N}(\bm{\upmu}_{\hat{\mathbf{Q}}_i}|\mathbf{Q}_i, \bm{\Sigma}_{\hat{\mathbf{Q}}_i}|\mathbf{Q}_i)
\end{align}
where $\bm{\upmu}_{\hat{\mathbf{Q}}_i}|\mathbf{Q}_i := \bm{\Sigma}_{\hat{\mathbf{Q}}_i,\mathbf{Q}_i}^{\sf T} \bm{\Sigma}_{\hat{\mathbf{Q}}_i}^{-1} \mathbf{Q}_i$ and $\bm{\Sigma}_{\hat{\mathbf{Q}}_i}|\mathbf{Q}_i = \bm{\Sigma}_{\hat{\mathbf{Q}}_i} - \bm{\Sigma}_{\hat{\mathbf{Q}}_i,\mathbf{Q}_i}^{\sf T} \bm{\Sigma}_{\mathbf{Q}_i} \bm{\Sigma}_{\hat{\mathbf{Q}}_i,\mathbf{Q}_i}$. 

So far the treatment of GPs has been relatively standard, however we are unable to observe $\mathbf{Q}_i$ and $\hat{\mathbf{Q}}_i$ directly. This problem is a common occurrence for latent state space models, for which MCMC sampling \cite{Frigola2013} or maximum a posteriori (MAP) \cite{Ko2011} has been applied. In this paper we apply MAP to obtain the required vectors $\mathbf{Q}_i$, for which we require the following likelihood based  on Equation \ref{eq:xcon} and Equation \ref{eq:y}:
\begin{subequations}
\begin{align}
    & \mathbf{y}_{k+1} \sim \mathcal{N}(\mathbf{H} \mathbf{x}_{k+1},\bm{\Sigma}_{\bm{\upnu}} + \mathbf{H} \bm{\Sigma_{\bm{\omega}}} \mathbf{H}^{\sf T})
\end{align}
\end{subequations}
where $\mathbf{x}_{k+1} = \int_{t_k}^{t_{k+1}} \mathbf{f}(\mathbf{x}(t),\mathbf{u}_k,\mathbf{q}(\mathbf{x}(t),\mathbf{u}_k)) dt + \mathbf{x}_k$ is dependent on the dynamics and crucially on the unknown function $\mathbf{q}(\cdot)$. 

Let $\mathbf{Q}$, $\hat{\mathbf{Q}}$, and $\bm{\Psi}$ refer to the joint $\mathbf{Q}_i$, $\hat{\mathbf{Q}}_i$ and $\bm{\Psi}_i$ respectively, i.e. $\mathbf{Q}=[\mathbf{Q}_1,\ldots,\mathbf{Q}_{n_{\mathbf{q}}}]$, $\hat{\mathbf{Q}}=[\hat{\mathbf{Q}}_1,\ldots,\hat{\mathbf{Q}}_{n_{\mathbf{q}}}]$, and $\bm{\Psi} = [\bm{\Psi}_1,\ldots,\bm{\Psi}_{n_{\mathbf{q}}}]$.

Based on the different likelihoods we can now write down the likelihood equation for the data: 
\begin{subequations}
\begin{align} \label{eq:y_likelihood}
     p(\mathbf{Y}|\mathbf{Q},\hat{\mathbf{Q}},\bm{\Psi},\mathbf{Z}) \propto & \quad p(\mathbf{Y}|\mathbf{Q},\hat{\mathbf{Q}},\mathbf{Z}) p(\hat{\mathbf{Q}}|\mathbf{Q},\bm{\Psi},\mathbf{Z}) \nonumber \\ & \times p(\mathbf{Q}|\bm{\Psi},\mathbf{Z}) p(\mathbf{Q}) p(\hat{\mathbf{Q}}) p(\bm{\Psi}) 
\end{align}
where the different likelihoods are given as follows:
\begin{align}
    & p(\mathbf{Y}|\hat{\mathbf{Q}},\mathbf{Z}) && =\prod_{j=1}^{N} \mathcal{N}(\mathbf{y}_{k+1}^{(j)};\hat{\mathbf{F}}(\mathbf{x}_k^{(j)},\mathbf{u}_k^{(j)},\hat{\mathbf{Q}}),\bm{\Sigma}_{\bm{\upnu}} + \mathbf{H} \bm{\Sigma_{\bm{\omega}}} \mathbf{H}^{\sf T}) \\
    & p(\hat{\mathbf{Q}}|\mathbf{Q},\bm{\Psi},\mathbf{Z}) &&  =\prod_{i=1}^{n_{\mathbf{q}}} \mathcal{N}(\hat{\mathbf{Q}}_i;\bm{\upmu}_{\hat{\mathbf{Q}}_i}|\mathbf{Q}_i, \bm{\Sigma}_{\hat{\mathbf{Q}}_i}|\mathbf{Q}_i) \\
    & p(\mathbf{Q}|\bm{\Psi},\mathbf{Z}) && =\prod_{i=1}^{n_{\mathbf{q}}} \mathcal{N}(\mathbf{Q}_i;\mathbf{0},\bm{\Sigma}_{\mathbf{Q}_i})
\end{align}
\end{subequations}
where $\hat{\mathbf{F}}(\mathbf{x}_k,\mathbf{u}_k,\hat{\mathbf{Q}})$ refers to a discredized state-space model, for which $\hat{\mathbf{Q}}$ represents the values of $\mathbf{q}(\cdot)$ at the discretizaton points. The likelihoods stated above can be understood as follows: $p(\mathbf{Y}|\mathbf{Q},\hat{\mathbf{Q}},\mathbf{Z})$ is the likelihood of the observed data given $\hat{\mathbf{Q}},\mathbf{Z}$, $p(\hat{\mathbf{Q}}|\mathbf{Q},\bm{\Psi},\mathbf{Z})$ is the likelihood of $\hat{\mathbf{Q}}$ given $\mathbf{Q},\bm{\Psi},\mathbf{Z}$, and lastly $p(\mathbf{Q}|\bm{\Psi},\mathbf{Z})$ refers to the likelihood of $\mathbf{Q}$ given $\bm{\Psi},\mathbf{Z}$.  

\begin{example}[Example discretization for MAP]
$\hat{\mathbf{F}}(\mathbf{x}_k,\mathbf{u}_k,\hat{\mathbf{Q}})$ can in general represent any valid discretization rule. Assume for example we apply the trapezium rule for discretization, then we obtain the following relation for the known input $\mathbf{z}_{k}^{(j)} = (\mathbf{x}_{k}^{(j)}, \mathbf{u}_k^{(j)})$:
\begin{equation} \label{eq:example_disc}
    \mathbf{x}^{(j)}_{k+1} = \mathbf{x}_{k}^{(j)} + 0.5 \delta_t \left(\mathbf{f}(\hat{\mathbf{x}}_{1}^{(j)},\mathbf{u}_k^{(j)},\hat{\mathbf{q}}_1^{(j)})+\mathbf{f}(\hat{\mathbf{x}}_{2}^{(j)},\mathbf{u}_k^{(j)},\hat{\mathbf{q}}_2^{(j)})\right)  
\end{equation}
where $\hat{\mathbf{x}}_i$ and $\hat{\mathbf{q}}_1^{(j)}$  refer to the $i$th state and $q$-value of the discretization rule and $\hat{\mathbf{F}}(\mathbf{x}_k,\mathbf{u}_k,\hat{\mathbf{Q}}) = \mathbf{x}^{(j)}_{k+1}$. These states $\hat{\mathbf{x}}_{1}^{(j)} = \mathbf{x}_{k}^{(j)}$ and $\hat{\mathbf{x}}_{2}^{(j)} = \mathbf{x}_{k+1}^{(j)}$, however note in general the initial- and end-point may not be part of the discretization points. Let the number of discretization points required per interval be given by $d_s$, such that for the trapezium rule above $d_s=2$. The corresponding matrices required for the MAP likelihood are given by $\hat{\mathbf{Q}}=[\hat{\mathbf{q}}_1^{(1)},\ldots,\hat{\mathbf{q}}_{d_s}^{(1)},\ldots,\hat{\mathbf{q}}_1^{(N)},\ldots,\hat{\mathbf{q}}_{d_s}^{(N)}]^{\sf T} \in \mathbb{R}^{(d_s N)\times n_{\mathbf{q}}}$ and \\ $\hat{\mathbf{Z}}=[\hat{\mathbf{z}}_1^{(1)},\ldots,\hat{\mathbf{z}}_{d_s}^{(1)},\ldots,\hat{\mathbf{z}}_1^{(N)},\ldots,\hat{\mathbf{z}}_{d_s}^{(N)}]^{\sf T} \in \mathbb{R}^{(d_s N)\times n_{\mathbf{q}}}$, where $\hat{\mathbf{z}}_{i}^{(j)}=(\hat{\mathbf{x}}_{i}^{(j)},\mathbf{u}_k^{(j)})$. For implicit integration rules as the one above either a Newton solver needs to be employed or the unknown values $\hat{\mathbf{x}}_{2}^{(j)}$ are added to the optimization variables with Equation \ref{eq:example_disc} as additional equality constraints for each training data-point.   
\end{example}

The remaining likelihoods $p(\mathbf{Q})$, $p(\hat{\mathbf{Q}})$, and $p(\bm{\Psi})$ are prior distributions of $\mathbf{Q}$, $\hat{\mathbf{Q}}$, and $\bm{\Psi}$ respectively. These are a helpful tool to avoid overfitting and can be used to easily integrate prior knowledge into the optimization problem, e.g. knowledge on the approximate magnitude of $\mathbf{Q}$. Refer to \cite{Ko2011} for examples on how priors can be used to incorporate prior knowledge on latent variables, such as $\mathbf{Q}$.     

The required values for $\mathbf{Q}$ and $\bm{\Psi}$ are then found by minimizing the negative log-likelihood of Equation \ref{eq:y_likelihood}:
\begin{equation}
  (\mathbf{Q}^*,\bm{\Psi}^*,\hat{\mathbf{Q}}^*) \in \underset{\mathbf{Q},\bm{\Psi},\hat{\mathbf{Q}}}{\text{argmin}} \, \mathcal{L}(\mathbf{Q},\bm{\Psi}) = -\log p(\mathbf{Y}|\mathbf{Q},\hat{\mathbf{Q}},\bm{\Psi},\mathbf{Z}) 
\end{equation}
where $\mathbf{Q}^*$, $\bm{\Psi}^*$, $\hat{\mathbf{Q}}^*$ are the required maximum a posteriori (MAP) estimates. 

In the following sections we assume that the GP has been fitted in this way such that we have MAP values $\mathbf{Q}^*$ and for $\bm{\Psi}^*$. The predictive distribution of $\mathbf{q}(\cdot)$ at an arbitrary input $\mathbf{z}=(\mathbf{x},\mathbf{u})$ is then given the dataset $\mathcal{D}=(\mathbf{Z},\mathbf{Q}^*)$:
\begin{subequations} \label{eq:GP_pred}
\begin{align}
    \mathbf{q}(\mathbf{z})|\mathcal{D} \sim \mathcal{N}(\bm{\upmu}_q(\mathbf{z};\mathcal{D}),\bm{\Sigma}_q(\mathbf{z};\mathcal{D}))
\end{align}

with
\begin{align} \label{GP_pred}
& \bm{\upmu}_q(\mathbf{z};\mathcal{D}) = [\mathbf{k}_1^{\sf T} \bm{\Sigma}^{-1}_{\mathbf{Q}_1} \mathbf{Q}_1^*,\ldots,\mathbf{k}_{n_{\mathbf{q}}}^{\sf T} \bm{\Sigma}^{-1}_{\mathbf{Q}_{n_{\mathbf{q}}}} \mathbf{Q}^*_{n_{\mathbf{q}}}]^{\sf T} \\
& \bm{\Sigma}_q(\mathbf{z};\mathcal{D}) = \\ & \qquad \quad \diag\left(\zeta^{2*}_1 + \sigma^{2*}_{\nu 1} - \mathbf{k}_1^{\sf T} \bm{\Sigma}^{-1}_{\mathbf{Q}_1} \mathbf{k}_1,\ldots,\zeta^{2*}_{n_{\mathbf{q}}} + \sigma^{2*}_{\nu n_{\mathbf{q}}} - \mathbf{k}_{n_{\mathbf{q}}}^{\sf T} \bm{\Sigma}^{-1}_{\mathbf{Q}_{n_{\mathbf{q}}}} \mathbf{k}_{n_{\mathbf{q}}} \right) \nonumber
\end{align}
\end{subequations}
where $\mathbf{k}_i=[k_i(\mathbf{q}_i^{in}(\mathbf{z}),\mathbf{q}_i^{in}(\mathbf{z}^{(1)}_k)),\ldots,k_i(\mathbf{q}_i^{in}(\mathbf{z}),\mathbf{q}_i^{in}(\mathbf{z}^{(N)}_k))]^{\sf T}$. In Figure \ref{fig:Prior_posterior} we illustrate a prior GP in the top graph and the posterior GP in the bottom graph.

\begin{figure}[H] \centering
   \includegraphics[width=0.95\textwidth]{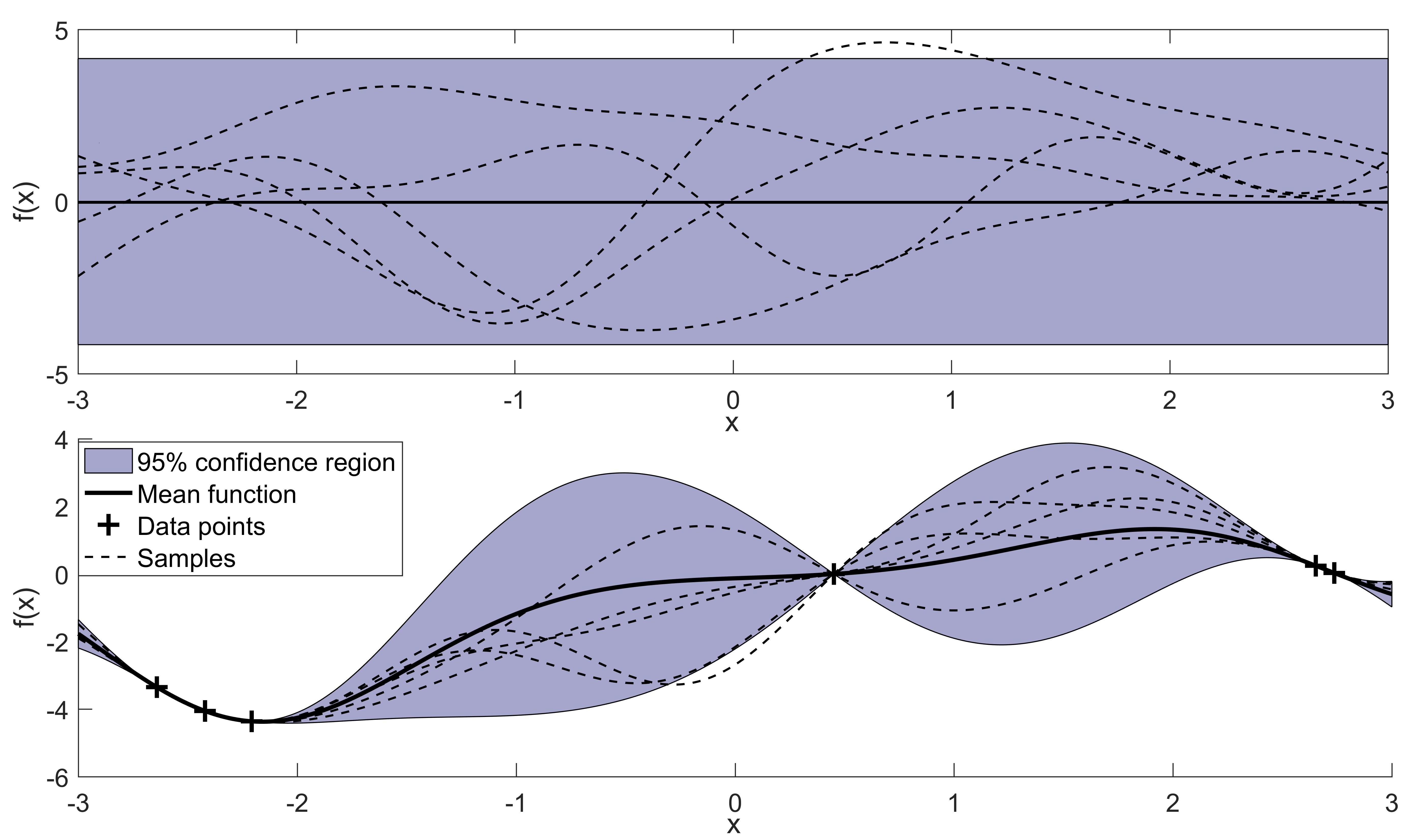}
  \caption{Illustration of a GP of a 1-dimensional function perturbed by noise. On the top the prior of the GP is shown, while on the bottom the Gaussian process was fitted to several data points to obtain the posterior.}
  \label{fig:Prior_posterior}
\end{figure}

\subsection{Hybrid Gaussian process model predictive control formulation} \label{sec:hybrid_GP_MPC}
In this section we define the NMPC optimal control problem (OCP) based on the GP hybrid \textit{nominal} model fitted in the previous section, where the \textit{nominal} model refers to the mean function in Equation \ref{eq:GP_pred}. The initial state for the GP hybrid NMPC formulation is assumed to be measured or estimated, and propagated forward using Equation \ref{eq:f_x}. The predicted states are exploited to optimize the objective subject to the tightened constraints. Let the corresponding optimization problem be denoted as $P_T\left(\bm{\upmu}_q(\cdot;\mathcal{D});\mathbf{x},k\right)$ for the current \textit{known} state $\mathbf{x}$ at discrete time $k$ based on the mean function $\bm{\upmu}_q(\cdot;\mathcal{D})$:
\begin{equation}
\begin{aligned} \label{eq:nominalMPC}
& \underset{\hat{\mathbf{U}}_{k:T-1}}{\text{minimize}} \quad \hat{V}_T(\mathbf{x},k,\hat{\mathbf{U}}_{k:T-1})  =  \sum_{j=k+1}^{T-1} \ell(\hat{\mathbf{x}}_j,\hat{\mathbf{u}}_j)  +  \ell_f(\hat{\mathbf{x}}_T) \\ 
& \text{subject to:}  \\
& \hat{\mathbf{x}}_{j+1} = \hat{\mathbf{F}}(\hat{\mathbf{x}}_j,\hat{\mathbf{u}}_j,\bm{\upmu}_q(\hat{\mathbf{z}}(t);\mathcal{D})), \quad \hat{\mathbf{z}}(t) = (\hat{\mathbf{x}}(t),\hat{\mathbf{u}}_j) \quad \forall j \in \{k,\ldots,T-1\} \\
& \hat{\mathbf{x}}_{j+1} \in  \overline{\mathbb{X}}_{j+1}, \quad \hat{\mathbf{u}}_j \in \mathbb{U} \quad \forall j \in \{k,\ldots,T-1\} \\
& \hat{\mathbf{x}}_k = \mathbf{x}
\end{aligned}
\end{equation}
where $\hat{\mathbf{x}}$, $\hat{\mathbf{u}}$, and $\hat{V}_T(\cdot)$ refers to the states, control inputs, and control objective of the MPC formulation, $\hat{\mathbf{U}}_{k:T-1} = [\hat{\mathbf{u}}_k,\ldots,\hat{\mathbf{u}}_{T-1}]^{\sf T}$, and $\overline{\mathbb{X}}_k$ is a tightened constraint set denoted by: $\overline{\mathbb{X}}_{k} = \{\mathbf{x} \in \mathbb{R}^{n_x} \; | \; g_i^{(k)}(\mathbf{x}) + b_i^{(k)} \leq 0, \, i=1,\ldots,n_g\}$. The variables $b_i^{(k)}$ represent so-called back-offs, which tighten the original constraints $\mathbb{X}_k$ defined in Equation \ref{eq:xcon}.

\begin{remark}[Objective in expectation]
Note the objective above in Equation \ref{eq:nominalMPC} aims to determine the optimal trajectory for the \textit{nominal} and not the expectation of the objective as defined in Equation \ref{eq:objdef}, since it is computationally expensive to obtain the expectation of a nonlinear function \cite{Hewing2017}. Further, the difference between the expectation and the nominal system is commonly marginal. 
\end{remark}

The NMPC algorithm solves $P_T\left(\bm{\upmu}_q(\cdot;\mathcal{D});\mathbf{x}_k,k\right)$ at each sampling time $t_k$ given the current state $\mathbf{x}_k$ to obtain an optimal control sequence:
\begin{align}
& \hat{\mathbf{U}}^*_{k:T-1}\left(\bm{\upmu}_q(\cdot;\mathcal{D});\mathbf{x}_k,k\right) =  [\hat{\mathbf{u}}^*_k\left(\bm{\upmu}_q(\cdot;\mathcal{D});\mathbf{x}_k,k\right),\ldots,\hat{\mathbf{u}}^*_{T-1}\left(\bm{\upmu}_q(\cdot;\mathcal{D});\mathbf{x}_k,k\right)]^{\sf T} 
\end{align}

Only the first optimal control action is applied to the plant at time $t_k$ before the same optimization problem is solved at time $t_{k+1}$ with a new state measurement $\mathbf{x}_{k+1}$. This procedure implicitly defines the following feedback control law, which needs to be repeatedly solved for each new measurement $\mathbf{x}_k$:
\begin{align} \label{eq:GP_control_policy}
\kappa(\bm{\upmu}_q(\cdot;\mathcal{D});\mathbf{x}_k,k) = \hat{\mathbf{u}}^*_k\left(\bm{\upmu}_q(\cdot;\mathcal{D});\mathbf{x}_k,k\right) 
\end{align}
It is explicitly denoted that the control actions depend on the GP hybrid model used. 

\begin{remark}[Full state feedback]
Note in the control algorithm we have assumed full state feedback, i.e. it is assumed that the full state can be measured without noise. This assumption can be dropped if required by introducing a suitable observer and introduced in the closed-loop simulations to account for this additional uncertainty.  
\end{remark}

\subsection{Closed-loop Monte Carlo sample} \label{sec:MC_sample}
In Equation \ref{eq:GP_control_policy} the control policy is stated, which is obtained by repeatedly solving the optimization problem in Equation \ref{eq:nominalMPC} with updated initial conditions. GPs are distribution over functions and hence a GP sample describes a \textit{deterministic} function. An example of this can be seen in Figure \ref{fig:Prior_posterior}, in which several GP samples are depicted. In this section we outline how MC samples of GPs can be obtained for a \textit{finite time horizon}, which each describe separate state trajectories according to Equation \ref{eq:f_x}. These are then exploited in the next section to tighten the constraints defined in the previous section. In general exact GP realizations cannot be obtained by any known approach, since generating such a sample would require sampling an infinite dimensional stochastic process. Instead, approximate approaches have been applied, such as spectral sampling \cite{Bradford2018d}. Exact samples of GP are however possible if the GP only needs to be evaluated at a \textit{finite number} of points. This is for example the case for discrete-time GP state space models as proposed in \cite{Conti2009,Umlauft2018}. We next outline this technique and show how this can be extended to the continuous-time case for hybrid GP models, in which discretization is applied. 

Assume we are given a state space model defined as in Section \ref{sec:prob_def} in Equation \ref{eq:f_x}, and a fitted GP model for $\mathbf{q}(\cdot)$ determined from Section \ref{sec:hybrid_GP_training}. The predictive distribution given the available data $\mathcal{D}=(\mathbf{Z},\mathbf{Y})$ is then given by Equation \ref{eq:GP_pred}. The aim here is to show how to obtain a sample of the state sequence, which can be repeated multiple times to obtain multiple possible state sequences. The initial condition $\mathbf{x}_0$ follows a known Gaussian distribution as defined in Equation \ref{eq:f_x}. Let $\bm{\mathcal{X}}^{(s)}=[\bm{\upchi}_0^{(s)},\ldots,\bm{\upchi}_T^{(s)}]^{\sf T}$ represent the state sequence of a GP realization $s$ and  $\bm{\upchi}_k^{(s)}$ the state of this realization at time $t=t_k$. Further, let the corresponding control actions at time $t=t_k$ be denoted by $\mathpzc{u}_k^{(s)}$. The control actions are assumed to be the result of the GP nominal NMPC feedback policy defined in Equation \ref{eq:GP_control_policy} and hence can be stated as:
\begin{equation}
   \mathpzc{u}_k^{(s)} =  \kappa(\bm{\upmu}_q(\cdot;\mathcal{D});\bm{\upchi}_k^{(s)},k)
\end{equation}

We denote the control actions over the time horizon $T$ jointly as $\mathcal{U}^{(s)}=[\mathpzc{u}_0^{(s)},\ldots,\mathpzc{u}_{T-1}^{(s)}]^{\sf T}=[\kappa(\bm{\upmu}_q(\cdot;\mathcal{D});\bm{\upchi}_{0}^{(s)},0),\ldots,\kappa(\bm{\upmu}_q(\cdot;\mathcal{D});\bm{\upchi}_{T-1}^{(s)},T-1)]$, which are different for each MC sample $s$ due to feedback.  

To obtain a sample of a state sequence we first need to sample the initial state $\mathbf{x}_0 \sim \mathcal{N}(\bm{\upmu}_{\bm{x}_0},\bm{\Sigma}_{\bm{x}_0})$ to attain the realization $\bm{\upchi}_0^{(s)}$. Thereafter, the next state is given by Equation \ref{eq:disc_x}, which is dependent on the fitted GP of $\mathbf{q}(\cdot)$. An exact approach to obtain an independent sample of a GP is as follows. Any time the GP needs to be evaluated at a certain point, the response at this point $\mathbf{q}(\cdot)$ is sampled according to the predictive distribution in Equation \ref{eq:GP_pred}. This sampled point is then part of the sampled function path, and hence the GP needs to be conditioned on it. This necessitates to treat this point as a noiseless pseudo \textit{training} point without changing the hyperparameters. Note if the sample path would return to the same evaluation point, it would then lead to the same output due to this conditioning procedure. Consequently, the sampled function is deterministic as expected.   

Furthermore, we also need to sample $\bm{\upomega}_k \sim \mathcal{N}(\mathbf{0},\bm{\Sigma}_{\bm{\upomega}})$ for each $k$. We refer to these realizations as $\mathpzc{w}_k^{(s)}$. We assume Equation \ref{eq:disc_x} has been adequately discretized, such that the GP of $\mathbf{q}(\cdot)$ needs to be evaluated at only a finite number of points. The state sequence for realization $s$ can then be given as follows:
\begin{equation}
    \bm{\upchi}_{k+1}^{(s)} = \hat{\mathbf{F}}(\bm{\upchi}_{k}^{(s)},\mathpzc{u}_k^{(s)},\mathcal{Q}_k^{(s)}) + \mathpzc{w}_k^{(s)} \quad \forall k \in \{1,\ldots,T\} 
\end{equation}
where $\mathcal{Q}_k^{(s)}$ are discretization points sampled from the GP following the procedure outlined above and $\hat{\mathbf{F}}(\cdot)$ represents the discretized version of Equation \ref{eq:disc_x}. 

\begin{example} 
We give an example here for the procedure above exploiting the trapezium rule for $\hat{\mathbf{F}}(\bm{\upchi}_{k}^{(s)},\mathpzc{u}_k^{(s)},\mathcal{Q}_k^{(s)})$. Note the covariance matrix and dataset of the GPs are updated recursively. Assume we are at time $k$ for MC sample $s$, and the covariance matrix of $q_i(\cdot)$ are given by $\bm{\Sigma}_{\mathbf{Q}_{ik}}^{(s)}$ with the updated data set $\mathcal{D}_k^{(s)}=(\mathbf{Z}_k^{(s)},\mathbf{Q}_k^{*(s)})$, where $\mathbf{Q}_k^{*(s)} = [\mathbf{Q}_{1k}^{*(s)},\ldots,\mathbf{Q}_{n_{\mathbf{q}}k}^{*(s)}]$ and $\mathbf{Z}_k^{(s)}=[\mathbf{z}^{(s1)},\ldots,\mathbf{z}^{(sN^k)}]^{\sf T}$ as in Section \ref{sec:hybrid_GP_training}. The dataset size $N^k = N + (k-1) \times d_s$, since at each time step $k$, $d_s$ discretization points are added to the dataset. 

Let the number of discretization points per time interval be given by $d_s$, for the trapezium rule $d_s=2$. The discretization points then follow the following distribution:
\begin{equation}
    \hat{\mathbf{Q}}_{ik}^{(s)} \in \mathbb{R}^{d_s} \sim \mathcal{N}(\bm{\upmu}_{\hat{\mathbf{Q}}_{ik}^{(s)}}|\mathbf{Q}_{ik}^{(s)}, \bm{\Sigma}_{\hat{\mathbf{Q}}_{ik}^{(s)}}|\mathbf{Q}_{ik}^{(s)})
\end{equation}
where $\bm{\upmu}_{\hat{\mathbf{Q}}_{ik}^{(s)}}|\mathbf{Q}_{ik}^{(s)} := \bm{\Sigma}_{\hat{\mathbf{Q}}_{ik}^{(s)},\mathbf{Q}_{ik}^{(s)}}^{\sf T} \bm{\Sigma}_{\hat{\mathbf{Q}}_{ik}^{(s)}}^{-1} \mathbf{Q}_{ik}^{(s)}$ and $\bm{\Sigma}_{\hat{\mathbf{Q}}_{ik}^{(s)}}|\mathbf{Q}_{ik}^{(s)} = \bm{\Sigma}_{\hat{\mathbf{Q}}_{ik}^{(s)}} - \bm{\Sigma}_{\hat{\mathbf{Q}}_{ik}^{(s)},\mathbf{Q}_{ik}^{(s)}}^{\sf T} \bm{\Sigma}_{\mathbf{Q}_{ik}^{(s)}} \bm{\Sigma}_{\hat{\mathbf{Q}}_{ik}^{(s)},\mathbf{Q}_{ik}^{(s)}}$, $[\bm{\Sigma}_{\hat{\mathbf{Q}}_{ik}^{(s)}}]_{lm} = k(\mathbf{q}_{i}^{in}(\hat{\mathbf{z}}^{(sl)}_k),\mathbf{q}_{i}^{in}(\hat{\mathbf{z}}^{(sm)}_k))$ for each pair $(l,m)\in\{ 1,\ldots,d_s\}^2$ and $[\bm{\Sigma}_{\hat{\mathbf{Q}}_{ik}^{(s)},\mathbf{Q}_{ik}^{(s)}}]_{lm} = k(\mathbf{q}_i^{in}(\hat{\mathbf{z}}^{(sl)}_k),\mathbf{q}_i^{in}(\mathbf{z}^{(sm)}_k))$ for each pair $(l,m)\in \{ 1,\ldots,N^k\} \times \{ 1,\ldots,d_s\}$.   

Firstly, we sample $d_s$ independent standard normally distributed $\bm{\upxi}_i \in \mathbb{R}^{n_{\mathbf{q}}} \sim \mathcal{N}(\mathbf{0},\mathbf{I})$ for each GP $i$. The sampled discretization points can then be expressed by:
\begin{equation} \label{eq:sample_points_q}
    \mathcal{Q}_{ik}^{(s)} = \bm{\upmu}_{\hat{\mathbf{Q}}_{ik}^{(s)}}|\mathbf{Q}_{ik}^{(s)} + \bm{\upxi}_i \cdot \bm{\Sigma}^{\frac{1}{2}}_{\hat{\mathbf{Q}}_{ik}^{(s)}}|\mathbf{Q}_{ik}^{(s)})
\end{equation}
where $\mathcal{Q}_k^{(s)} \in \mathbb{R}^{d_s} \sim \mathcal{N}(\bm{\upmu}_{\hat{\mathbf{Q}}_{ik}^{(s)}}|\mathbf{Q}_{ik}^{(s)},\bm{\Sigma}_{\hat{\mathbf{Q}}_{ik}^{(s)}}|\mathbf{Q}_{ik}^{(s)})$.

Once $\mathcal{Q}_k^{(s)}$ has been sampled we arrive at the next state $k+1$ for the MC as follows:
\begin{equation} \label{eq:example_disc2}
    \bm{\upchi}^{(s)}_{k+1} = \bm{\upchi}_{k}^{(s)} + 0.5 \delta_t \left(\mathbf{f}(\hat{\upchi}_{1}^{(j)},\mathbf{u}_k^{(j)},\mathcal{Q}_{1k}^{(s)})+\mathbf{f}(\hat{\upchi}_{2}^{(j)},\mathbf{u}_k^{(j)},\mathcal{Q}_{2k}^{(s)})\right)  
\end{equation}
where $\hat{\upchi}_{1}^{(j)}=\bm{\upchi}_{k}^{(s)}$ and $\hat{\upchi}_{2}^{(j)}=\bm{\upchi}^{(s)}_{k+1}$, since for the trapezium rule the discretization points coincide with the initial and the end-points, which is not true for other discretization rules. The value of the inputs for the discredizations points are consequently given by $\hat{\mathbf{z}}^{(sl)}_k = (\hat{\upchi}_{l}^{(s)},\mathpzc{u}_k^{(s)})$. For implicit integration rules as the one above a Newton solver needs to be employed using Equations \ref{eq:sample_points_q} and \ref{eq:example_disc2}. Equation \ref{eq:sample_points_q} is required due to the dependency of $\hat{\mathbf{z}}^{(sl)}_k$ on $\hat{\upchi}_{l}^{(j)}$.  

Lastly, the data matrices for the particular MC sample need to be updated as follows:
\begin{subequations}
\begin{align}
    & \mathbf{Z}_{k+1}^{(s)}                = [\mathbf{Z}_k^{(s)},\hat{\mathbf{z}}^{(s 1)}_k,\ldots,\hat{\mathbf{z}}^{(s d_s)}_k]^{\sf T} \\
    & \mathbf{Q}^{*(s)}_{k+1} = [\mathbf{Q}^{*(s)\sf T}_{k},\mathcal{Q}_k^{(s)\sf T}]^{\sf T} \\
    & [\bm{\Sigma}_{\mathbf{Q}_{ik}}]_{lm} = k_i(\mathbf{q}_i^{in}(\mathbf{z}^{(sl)}_k),\mathbf{q}_i^{in}(\mathbf{z}^{(sm)}_k)) + \sigma^2_{\nu i} \delta_{l,m} \quad \forall (l,m)\in\{ 1,\ldots,N^{k+1}\}^2
\end{align}
\end{subequations}
\end{example}

Repeating this procedure multiple times then gives us multiple MC samples of the state sequence $\bm{\mathcal{X}}^{(s)}$. The aim then is to use the information obtained from these sequences to iteratively tighten the constraints for the GP NMPC problem in Equation \ref{eq:nominalMPC} to obtain the probabilistic constraint satisfaction required from the initial problem definition in Section \ref{sec:prob_def}. 
 
\subsection{Probabilistic constraint tightening}
This section outlines how to systemically tighten the constraints based on MC samples using the procedure outlined in the previous chapter. Firstly define the function $C(\cdot)$, which is a single-variate random variable that represents the satisfaction of the joint chance constraints:
\begin{subequations} \label{eq:c_def}
\begin{align}
    & C(\mathbf{X}) = \inf_{(j,k) \in \{1,\ldots,n_g\} \times \{0,\ldots,T\}} {g_j^{(k)}(\mathbf{x}_k)} \\
    & F_{C(\mathbf{X})} = \mathbb{P}\left\{C(\mathbf{X}) \leq 0\right\} = \mathbb{P} \left\{ \bigcap^T_{k=0} \{ \mathbf{x}_k \in \mathbb{X}_k \}    \right\}
\end{align}
\end{subequations}
where $\mathbf{X} = [\mathbf{x}_0,\ldots,\mathbf{x}_T]^{\sf T}$ defines a state sequence, and $\mathbb{X}_k = \{ \mathbf{x} \in \mathbb{R}^{n_{\mathbf{x}}} \mid g_j^{(k)}(\mathbf{x}) \leq 0, j=1,\ldots,n_g \}$. 

The evaluation of the probability in Equation \ref{eq:c_def} is generally intractable, and instead a non-parametric sample approximation is applied, known as the \textit{empirical cumulative distribution function} (ecdf). Assuming we are given $S$ MC samples of the state trajectory $\mathbf{X}$ and hence of $C(\mathbf{X})$, the ecdf estimate of the probability in Equation \ref{eq:c_def} can be defined as follows:
\begin{equation} \label{eq:approx_joint}
    F_{C(\mathbf{X})} \approx \hat{F}_{C(\mathbf{X})} = \frac{1}{S} \sum_{s=1}^{S} \mathbf{1}\{C(\bm{\mathcal{X}}^{(s)}) \leq 0\} 
\end{equation}
where $\bm{\mathcal{X}}^{(s)}$ is the $s$-th MC sample and $\hat{F}_{C(\mathbf{X})}$ is the ecdf approximation of the true probability $F_{C(\mathbf{X})}$. 

The accuracy of the ecdf in Equation \ref{eq:approx_joint} significantly depends on the number of samples used and it is therefore paramount to account for the residual uncertainty of this sample approximation. This problem has been previously studied in statistics, for which the following probabilistic lower bound has been proposed known as ``\textit{exact confidence bound}'' \cite{Clopper1934}:
\begin{theorem}[Confidence interval for empirical cumulative distribution function] Assume we are given a value of the ecdf, $\hat{\beta} = \hat{F}_{C(\mathbf{X})}$, as defined in Equation \ref{eq:approx_joint} based on $S$ independent samples of $C(\mathbf{X})$, then the true value of the cdf, $\beta = F_{C(\mathbf{X})}$, as defined in Equation \ref{eq:c_def} has the following lower confidence bounds:
\begin{align} \label{eq:beta_lower_bound}
    & \mathbb{P}\left\{ \beta \geq \hat{\beta}_{lb} \right\} \geq 1-\alpha, && \hat{\beta}_{lb} = \textup{betainv}\left(\alpha,S + 1 - S \hat{\beta}, S \hat{\beta} \right)
\end{align}
\end{theorem}
\begin{proof}
The proof uses standard results in statistics and can be found in \cite{Clopper1934,Streif2014}. \qed
\end{proof}

In other words the probability that the probability defined in Equation \ref{eq:c_def}, $\beta$, exceeds $\hat{\beta}_{lb}$ is greater than $1-\alpha$. Consequently, for small $\alpha$ $\hat{\beta}_{lb}$ can be seen as a conservative lower bound of the true probability $\beta$ accounting for the statistical error introduced through the \textit{finite} sample approximation. Based on the definition of $C(\mathbf{X})$ and the availability of $S$ closed-loop MC simulations of the state sequence $\mathbf{X}$, assume we are given a value for $\hat{\beta}_{lb}$ according to Equation \ref{eq:beta_lower_bound} with a confidence level of $1-\alpha$, then the following Corollary holds:  
\begin{corollary}[Feasibility probability]
Assuming the stochastic system in Equation \ref{eq:f_x} is a correct description of the uncertainty of the system including the fitted GP and ignoring possible inaccuracies due to discretization errors, and given a value of the lower bound $\hat{\beta}_{lb} \geq 1-\epsilon$ defined in Equation \ref{eq:beta_lower_bound} with a confidence level of $1-\alpha$, then the original chance constraint in Equation \ref{eq:xcon} holds true with a probability of at least $1-\alpha$.  
\end{corollary}
\begin{proof}
The realizations of possible state sequences described in Section \ref{sec:MC_sample} are exact within an arbitrary small discretization error and therefore these $S$ independent state trajectories $\bm{\mathcal{X}}$ provide a valid lower bound $\hat{\beta}_{lb}$ from Equation \ref{eq:beta_lower_bound} to the true cdf value $\beta$. If $\hat{\beta}_{lb}$ is greater than or equal to $1-\epsilon$, then the following probabilistic bound holds on the true cdf value $\beta$ according to Theorem 1: $\mathbb{P}\left\{\beta \geq \hat{\beta}_{lb} \geq  1 - \epsilon \right\} \geq 1-\alpha$, which in other words means that $\beta = \mathbb{P}\left\{ C(\mathbf{X} \leq 0) \right\} \geq 1-\epsilon$ with a probability of at least $1-\alpha$. \qed 
\end{proof}    
     
Now assume we want to determine back-off values for the nominal GP NMPC algorithm in Equation \ref{eq:nominalMPC}, such that $\beta_{lb}$ is equal to $1-\epsilon$ for a chosen confidence level $1-\alpha$. This then in turn guarantees the satisfaction of the original chance constraint with a probability of at least $1-\alpha$. The update rule to accomplish this has two steps: Firstly an approximate constraint set is defined and secondly this set is iteratively adjusted. The approximate constraint set should reflect the difference of the constraint values for the state sequence of the nominal MPC model and the constraint values of possible state sequence realizations of the \textit{real} system in Equation \ref{eq:f_x}. The back-offs are first set to zero and $S$ MC samples are run according to Section \ref{sec:MC_sample}. Now assume we aim to obtain back-off values that imply satisfaction of individual chance constraints as follows to attain an approximate initial constraint set:    
\begin{align} \label{eq:chance_implication}
    & g_j^{(k)}(\overline{\bm{\upchi}}_k) + b_j^{(k)} = 0 \implies \mathbb{P}\left\{g_j^{(k)}(\bm{\upchi}_k) \leq 0 \right\} \geq 1 - \delta 
\end{align}
where $\delta$ is a tuning parameter and should be set to a reasonably low value and $\overline{\bm{\upchi}}_k$ refers to states according to the \textit{nominal} trajectory as defined in Section \ref{sec:MC_sample}.     

It is proposed in \cite{Paulson2018} to exploit the inverse ecdf to fulfill the requirement given in Equation (\ref{eq:chance_implication}) using the $S$ MC samples available. The back-offs can then be stated as:
\begin{equation} \label{eq:update_rule}
\tilde{b}_j^{(k)} = \hat{F}_{g_j^{(k)}}^{-1}(1-\delta) - g_j^{(k)}(\overline{\bm{\upchi}}_k) \quad \forall (j,k) \in \{1,\ldots,n_g^{(k)}\} \times \{1,\ldots,T\} 
\end{equation}
where $\hat{F}_{g_j^{(t)}}^{-1}$ denotes the inverse of the ecdf given in Equation \ref{eq:approx_joint} and $\tilde{b}_j^{(t)}$ refers to these initial back-off values. The inverse of an ecdf can be determined by the quantile values of the $S$ constraint values from the MC samples with cut-off probability $1-\delta$.     

This first step gives us an initial constraint set that depends on the difference between the nominal prediction $\overline{\bm{\upchi}}_k$ as used in the MPC and possible state sequences according to the MC simulations. The parameter $\delta$ in this case is only a tuning parameter to obtain the initial back-off values. 

In the next step these back-off values are modified using a \textit{back-off factor} $\gamma$:
\begin{equation} \label{eq:back_off_factor_def}
   b_j^{(k)} = \gamma \tilde{b}_j^{(k)} \quad \forall (j,k) \in \{1,\ldots,n_g^{(k)}\} \times \{1,\ldots,T\}
\end{equation}

A value of $\gamma$ is sought for which the lower bound $\beta_{lb}$ is equal to $1-\epsilon$ to obtain the required chance constraint satisfaction in Equation \ref{eq:xcon}, which can be formulated as a root finding problem:
\begin{equation} \label{eq:root}
    h(\gamma) = \hat{\beta}_{lb}(\gamma) - (1 - \epsilon)
\end{equation}
where the aim is to determine a value of $\gamma$, such that $h(\gamma)$ is approximately zero. $\hat{\beta}_{lb}(\gamma)$ refers to the implicit dependence of $\hat{\beta}_{lb}$ on the $S$ MC simulations resulting from the tightened constraints of the nominal GP NMPC algorithm according to Equation \ref{eq:back_off_factor_def}.

In other words the back-off values of the NMPC are adjusted until they return the required chance constraint satisfaction in Equation \ref{eq:xcon}. To drive $h(\gamma)$ to zero we employ the $\textit{bisection technique}$ \cite{Beers2007}, which seeks the root of a function in an interval $a_{\gamma}$ and $b_{\gamma}$, such that $h(a_{\gamma})$ and $h(b_{\gamma})$ have opposite signs. It is expected that a too high value of the back-off factor leads to a highly conservative solution with a positive sign of $h(b_{\gamma})$, while a low value of the back-off factor often results in negative values of $h(b_{\gamma})$. In our algorithm the initial $a_{\gamma}$ is set to zero to evaluate $\tilde{b}_j^{(k)}$ in the first step. The \textit{bisection method} repeatedly bisects the interval, in which the root is contained. The output of the algorithm are the required back-offs in $n_b$ back-off iterations. The overall procedure to attain the back-offs in Algorithm \ref{alg:back_off_algorithm}.

\begin{algorithm2e}[H] \label{alg:back_off_algorithm}
 \caption{Back-off iterative updates}
\Input{$\bm{\upmu}_{\mathbf{x}_0}$, $\bm{\Sigma}_{\mathbf{x}_0}$, $\bm{\upmu}_q(\mathbf{z};\mathcal{D})$, $\bm{\Sigma}_q(\mathbf{z};\mathcal{D})$, $\mathcal{D}$, $T$, $V_T(\mathbf{x},k,\hat{\mathbf{U}}_{k:T-1})$, $\mathbb{X}_k$, $\mathbb{U}_k$, $\epsilon$, $\alpha$, $\delta$, learning, S, $n_b$}

\Initialize{Set all $b_j^{(k)}=0$ and $\delta$ to some reasonable value, set $a_{\gamma}=0$ and $b_{\gamma}$ to some reasonably high value, such that $b_{\gamma} - (1-\epsilon)$ has a positive sign.}

\For{$n_b$ back-off iterations}{
\uIf{$n_b > 0$}{$c_{\gamma} := (a_{\gamma} + b_{\gamma})/2$ \\ $b_j^{(t)} := c_{\gamma} \tilde{b}_j^{(t)} \quad (j,t) \in \{1,\ldots,n_g^{(t)}\} \times \{1,\ldots,T\}$} 

\vspace{5pt}
Define GP NMPC in Equation \ref{eq:nominalMPC} with back-offs $b_j^{(t)}$ \\
Run $S$ MC simulations to obtain $\bm{\mathcal{X}}^{(s)}$ using the GP NMPC policy with updated back-offs  \vspace{5pt}

$\hat{\beta} := \hat{F}_{C(\bm{\mathcal{X}}^{(s)})} = \frac{1}{S} \sum_{s=1}^{S} \mathbf{1}\{C(\bm{\mathcal{X}}^{(s)}) \leq 0\}$  \\
$\hat{\beta}_{lb} := \text{betainv}\left(\alpha,S + 1 - S \hat{\beta}, S \hat{\beta} \right)$ 

\uIf{$nb = 0$}{
$\tilde{b}_j^{(t)} = \hat{F}_{g_j^{(t)}}^{-1}(\delta) - g_j^{(t)}(\overline{\bm{\upchi}}_t) \, \forall (j,t) \in \{1,\ldots,n_g^{(t)}\} \times \{1,\ldots,T\}$ \\
$\hat{\beta}_{lb}^{a_{\gamma}} := \hat{\beta}_{lb} - (1-\epsilon) $} 

\Else{$\hat{\beta}_{lb}^{c_{\gamma}} := \hat{\beta}_{lb} - (1-\epsilon)$ \\ 
\uIf{ $\text{sign}(\hat{\beta}_{lb}^{c_{\gamma}}) = \text{sign}(\hat{\beta}_{lb}^{a_{\gamma}})$}{$a_{\gamma} := c_{\gamma}$ \\
$\hat{\beta}_{lb}^{a_{\gamma}} := \hat{\beta}_{lb}^{c_{\gamma}}$}
\Else{$b_{\gamma} := c_{\gamma}$}}}

\Output{$b_j^{(t)} \quad \forall (j,t) \in \{1,\ldots,n_g^{(t)}\} \times \{1,\ldots,T\}, \, \hat{\beta}_{lb}$}
\end{algorithm2e}

\subsection{Algorithm}
A summary of the overall algorithm proposed in this paper is given in this section. As first step the problem needs to be specified following the problem definition in Section \ref{sec:prob_def}. From the available data the GP hybrid model needs to be trained as outlined in Section \ref{sec:hybrid_GP_training}. Thereafter, the back-offs are determined \textit{offline} iteratively following Algorithm \ref{alg:back_off_algorithm}. These back-offs then define the tightened constraint set for the GP NMPC feedback policy \textit{online}, which is run online to solve the problem initially outlined. An overall summary can be found in Algorithm \ref{alg:algorithm_summary}. 

\begin{algorithm2e}[H] \label{alg:algorithm_summary}
 \caption{Back-off GP NMPC}
 \textit{Offline Computations}
 \begin{enumerate}
\item{Build GP hybrid model from data-set $\mathcal{D}=(\mathbf{Z},\mathbf{Y})$ as shown in Section \ref{sec:hybrid_GP_training}.}
\item{Choose time horizon $T$, initial condition mean $\bm{\upmu}_{\mathbf{x}_0}$ and covariance $\bm{\Sigma}_{\mathbf{x}_0}$, measurement covariance matrix $\bm{\Sigma}_{\bm{\upnu}}$, disturbance covariance matrix $\bm{\Sigma}_{\bm{\upomega}}$, stage costs $\ell$ and $\ell_f$, constraint sets $\mathbb{X}_k, \mathbb{U}_k$ $\forall k \in \{1,\ldots,T\}$, chance constraint probability $\epsilon$, ecdf confidence $\alpha$, tuning parameter $\delta$, the number of back-off iterations $n_b$, and the number of Monte Carlo simulations $S$ to estimate the back-offs.}
\item{Determine explicit back-off constraints using Algorithm \ref{alg:back_off_algorithm}.}
\item{Check final probabilistic value $\hat{\beta}_{lb}$ from Algorithm \ref{alg:back_off_algorithm} if it is close enough to $\epsilon$.}
\end{enumerate} 
\textit{Online Computations} \\ \For{$k=0,\ldots,T-1$}{
\begin{enumerate}
\item{Solve the MPC problem in Equation \ref{eq:nominalMPC} with the tightened constraint set from the \textit{Offline Computations}.}
\item{Apply the first control input of the optimal solution to the \\ real plant.} 
\item{Measure the current state $\mathbf{x}_k$.}
\end{enumerate}}
\end{algorithm2e}

\section{Case study} \label{sec:case_study}
The case study is based on a semi-batch reaction for the production of fatty acid methyl ester (FAME) from microalgae, which is considered a promising renewable feedstock to meet the growing global energy demand. FAME is the final product of this process, which can be employed as biodiesel \cite{DelRioChanona2018}. We exploit a simplified dynamic model to verify the hybrid GP NMPC algorithm proposed in this paper. The GP NMPC has an economic objective, which is to maximize the FAME (biodiesel) concentration for the final batch product subject to two path constraints and one terminal constraint.   
 
\subsection{Semi-batch bioreactor model}
The simplified dynamic system consists of four ODEs describing the evolution of the concentration of biomass, nitrate, nitrogen quota, and FAME. We assume a fixed volume fed-batch. The balance equations can be stated as follows \cite{DelRioChanona2018}:   
\begin{align} \label{eq:case_study_ode}
    & \frac{dC_X}{dt} = 2\mu_m(I_0,C_X)\left(1 - \frac{k_q}{q}\right) \left(\frac{N}{N+K_N}\right)  C_X - \mu_d C_X, \, C_X(0) = {C_X}_0            \nonumber \\
    & \frac{dC_N}{dt} = -\mu_N \left(\frac{C_N}{C_N+K_N}\right) C_X + F_N, \quad C_N(0) = {C_N}_0 \\
    & \frac{dq}{dt} = \mu_N \left(\frac{C_N}{C_N+K_N}\right) - \mu_m(I_0,C_X) \left(1 - \frac{k_q}{q}\right) q, \quad q(0) = q_0                \nonumber \\
    & \frac{d\textit{FA}}{dt} = \mu_m(I_0,C_X) (\theta' q - \epsilon' \textit{FA}) \left(1 - \frac{k_q}{q}\right)  \nonumber \\ \nonumber
    & \quad \quad - \gamma' \mu_N \left(\frac{C_N}{C_N+K_N}\right) C_X, \quad \textit{FA}(0) = \textit{FA}_0  
\end{align}
where $C_X$ is the concentration of biomass in gL$^{-1}$, $C_N$ is the nitrate concentration in mgL$^{-1}$, $q$ is the dimensionless intracellular nitrogen content (nitrogen quota), and $\textit{FA}$ is the concentration of FAME (biodiesel) in gL$^{-1}$. Control inputs are given by the incident light intensity ($I_0$) in $\upmu \text{mol.m}^{-2}$.s$^{-1}$ and nitrate inflow rate ($F_N$) in mg.L$^{-1}$.h$^{-1}$. The state vector is hence given by $\mathbf{x}=[C_X,C_N,q,\textit{FA}]^{\sf T}$ and the input vector by $\mathbf{u}=[I_0,F_N]^{\sf T}$. The corresponding initial state vector is given by $\mathbf{x}_0=[{C_X}_0,{C_N}_0,q_0,\textit{FA}_0]^{\sf T}$. The remaining parameters can be found in Table \ref{tab:par_casestudy} taken in part from \cite{DelRioChanona2018}.  

\begin{table}[H] 
\centering
\caption{Parameter values for ordinary differential equation system in Equation \ref{eq:case_study_ode}.}
\begin{tabular}{*{3}{l}} \label{tab:par_casestudy}
Parameter & Value & Units \\
\hline
$\mu_M$             & 0.359                & $\text{h}^{-1}$                           \\
$\mu_d$             & 0.004                & $\text{h}^{-1}$                           \\
$k_q$               & 1.963                & mg.g$^{-1}$                               \\
$\mu_N$             & 2.692                & mg.g$^{-1}$.h$^{-1}$                      \\
$K_N$               & 0.8                  & mg.L$^{-1}$                               \\
$k_s$               & 91.2                 & $\upmu \text{mol.m}^{-2}\text{.s}^{-1}$   \\
$k_i$               & 100.0                & $\upmu \text{mol.m}^{-2}\text{.s}^{-1}$   \\
$\alpha'$           & 196.4                & L.mg$^{-1}$.m$^{-1}$                      \\
$\theta'$           & 6.691                & -                                         \\
$\gamma'$           & 7.53 $\times 10^3$   & -                                         \\
$\epsilon'$         & 0.01                 & -                                         \\
$\tau'$             & 1.376                & -                                         \\
$\delta'$           & 9.904                & -                                         \\
$\phi'$             & 16.89                & -                                         \\
$\beta'$            & 0.0                  & m$^{-1}$                                  \\
$L$                 & 0.0044               & m               
\end{tabular}
\end{table}

The function $\mu_m(I_0,C_X)$ describes the complex effects of light intensity on the biomass growth, which we assume to be unknown in this study. This helps simplify the model significantly, since these effects are dependent on the distance from the light source and hence would lead to a partial differential equation (PDE) model if modelled by first principles. The actual function can be given as follows to obtain values to train the hybrid GP:
\begin{align}
    \mu_m(I_0,C_X) = \frac{\mu_M}{L}  \int_{z=0}^L \left( \frac{I(z,I_0,C_X)}{I(z,I_0,C_X)+k_s+\frac{I(z,I_0,C_X)^2}{k_i}} \right) dz 
\end{align}
where $I(z,I_0,C_X) = I_0 \exp\left(-(\alpha' C_X + \beta')z\right)$, $z$ is the distance from the light source in m, and $L$ is the reactor width. 

\subsection{Problem set-up}
The problem has a time horizon $T=12$ with a batch time of $480$h, and hence a sampling time of $40$h. Next we state the objective and constraint functions according to the general problem definition in Section \ref{sec:prob_def} based on the dynamic system in Equation \ref{eq:case_study_ode}. 

Measurement noise covariance matrix $\bm{\Sigma}_{\bm{\upnu}}$ and disturbance noise matrix $\bm{\Sigma}_{\bm{\upomega}}$ are defined as:
\begin{subequations}
\begin{align} \label{eq:noise_matrices_case}
    & \bm{\Sigma}_{\bm{\upnu}}    = 10^{-4} \times  \diag\left(2.5^2,800^2,500^2,3000^2\right)  \\
    & \bm{\Sigma}_{\bm{\upomega}} = 10^{-4} \times \diag\left(0.1^2,200^2,10^2,100^2\right)
\end{align}
\end{subequations}

The mean and covariance of the initial condition are set to:
\begin{equation}
    \bm{\upmu}_{\mathbf{x}_0} = [0.4,0,150,0]^{\sf T}, \bm{\Sigma}_{\mathbf{x}_0} = 10^{-3} \times \diag(0.2^2,0,100^2,0)
\end{equation}

The aim of the control problem is to maximize the amount of biodiesel in the final batch with a penalty  on the chance of control actions. The corresponding stage and terminal costs can be given as:
\begin{align}
    \ell(\mathbf{x}_t,\mathbf{u}_t) = \bm{\Delta}_{\mathbf{u}_t}^{\sf T} \mathbf{R} \bm{\Delta}_{\mathbf{u}_t}, \quad \ell_f(\mathbf{x}_T) = -\textit{FA}_T 
\end{align}
where $\bm{\Delta}_{\mathbf{u}_t} = \mathbf{u}_t -\mathbf{u}_{t-1}$ and $\mathbf{R}= 5 \times 10^{-3} \times \diag(1/400^2,1/40^2)$. The objective is then defined by Equation \ref{eq:objdef}.  

There are two path constraints. Firstly, the nitrate is constrained to be below $800$mg/L. Secondly, the ratio of nitrogen quota $q$ to biomass may not exceed 0.011 for high density biomass cultivation. These are then defined as:  
\begin{subequations}
\begin{align}
    & g_1^{(t)} = {C_N}_t - 800 \leq 0 && \forall t \in \{0,\ldots,T\} \\
    & g_2^{(t)} = q_t - 0.011 {C_X}_t \leq 0 && \forall t \in \{0,\ldots,T\}
\end{align}
\end{subequations}

Further, the nitrate should reach a concentration below $150$mg/L for the final batch. This constraints can be stated as:
\begin{equation}
    g_3^{(T)}(\mathbf{x}_T) = {C_N}_T - 200 \leq 0, \, g_3^{(t)}(\mathbf{x}_t) = 0 \, \forall t \in \{0,\ldots,T-1\} 
\end{equation}

The control inputs light intensity and nitrate inflow rate are subject to the following box constraints:
\begin{subequations}
\begin{align}
    & 120 \leq I_t     \leq 300 && \forall t \in \{0,\ldots,T\} \\
    & 0   \leq {F_N}_t \leq 10  && \forall t \in \{0,\ldots,T\}
\end{align}
\end{subequations}

The priors were set to the following values:
\begin{subequations}
\begin{align}
    & p(\mathbf{Q})       = \mathcal{N}(-6 \times \mathbf{1}, 50 \times \mathbf{I})     \\
    & p(\hat{\mathbf{Q}}) = \mathcal{N}(-6 \times \mathbf{1}, 50 \times \mathbf{I})     \\
    & p(\bm{\Psi})        = \mathcal{N}([\mathbf{0},5\times10^{-3}]^{\sf T},\diag(20 \times \mathbf{I},1\times10^{-6}))
\end{align}
\end{subequations}

Maximum probability of violation was to $\epsilon=0.1$. To compute the back-offs a total of $S=1000$ MC iterations are employed for each iteration according with $\delta=0.05$ and $\alpha=0.01$. The number of back-off iterations was set to $n_b = 14$.    

\subsection{Implementation and initial dataset generation}
The discretization rule used for the MAP fit, for the GP MC sample, and for the GP NMPC formulation exploits direct collocation with $4$th order polynomials with the Radau collocation points. The MAP optimization problem and the GP NMPC optimization problem are solved using Casadi \cite{Andersson2018} to obtain the gradients of the problem using automatic differentiation in conjunction with IPOPT \cite{Wachter2006}. IDAS \cite{Hindmarsh2005} is utilised to simulate the  "real" plant. The input dataset  $\mathbf{Z}$ was designed using the Sobol sequence \cite{Sobol2001} for the entire input data in the range $\mathbf{z}_i \in [0,3] \times [0,800] \times [0,600] \times [0,3500] \times [120,300] \times [0,10]$. The ranges were chosen for the data to cover the expected operating region. The outputs $\mathbf{Y}$ were then obtained from the IDAS simulation of the system perturbed by Gaussian noise as defined in the problem setup. 

\section{Results and discussions} \label{sec:results}
Firstly, the accuracy of the proposed hybrid GP model is verified by creating 1000 random datapoints. For these we calculate the absolute prediction error and the absolute error over the standard deviation, which gives an indication on the accuracy of the uncertainty measure provided by the GP. These results are summarized in Figure \ref{fig:Boxplot_hybrid}. For comparison purposes three cases of the GP NMPC algorithm are compared. Firstly, we run the above case study using 30 datapoints and 50 datapoints. In addition, we compare this with the previously proposed GP NMPC algorithm in \cite{Bradford2019e} that aims to model the dynamic state space equations using GPs using 50 datapoints. Lastly, these three cases are further compared to their \textit{nominal} variations, i.e. setting all back-offs in the formulations to zero. The results of these runs are highlighted in Figures \ref{fig:Con_MC_30_hybrid}-\ref{fig:Plant_constraint_plots_nominal} and in Table \ref{tab:probability_OCP_time}. From these results we can draw the following conclusions:
\begin{itemize}
    \item{From Figure 2 we can firstly see in the first graph that the median absolute error decreases significantly going from a dataset size of $30$ to $50$, which is as expected. Overall the hybrid model predictions seem reasonably well. The GP error measure can be tested by dividing the absolute error by the standard deviation, for which the vast majority of values should be within approximately a range of 0 to 3. A value above $3$ has a chance of $99.4\%$ of occurrence according to the underlying Gaussian distribution. For $N=30$ we observe no value above $3$, while for $N=50$ we observed $1.1\%$. It can therefore be said that the error measure for $N=30$ is more conservative, but both seem to show reasonable behaviour.}
    \item{From Figures \ref{fig:Con_MC_30_hybrid}-\ref{fig:Con_MC_50_nonhybrid} it can be seen that the hybrid approaches both lead to generally good solutions, while the non-hybrid approach is unable to deal with the spread of the trajectories for constraint $g_2$. The resulting  Further, it can be seen that the uncertainty of GP hybrid 50 is less than GP hybrid 30 from the significantly smaller spread of constraint $g_2$, which is as expected given the observations from Figure \ref{fig:Boxplot_hybrid}.}
    \item{Figure \ref{fig:Objective_plots} illustrates the better performance of GP hybrid 50 over GP hybrid obtaining a nearly $40\%$ increase in the objective on average. This is due to two reasons: Firstly more data leads to better decisions on average and secondly due to lower uncertainty the GP hybrid 50 is less conservative than GP hybrid 30. Lastly, GP non-hybrid 50 achieves high objective values by violating the second constraint $g_2$ by a substantial amount.}
    \item{Figures \ref{fig:Plant_constraint_plots} and \ref{fig:Plant_constraint_plots_nominal} show that the \textit{nominal} approach ignoring back-offs leads to constraint violations for all GP NMPC variations, while with back-offs the two hybrid approaches remain feasible throughout. GP non-hybrid 50 overshoots the constraint by a huge amount due to the NMPC becoming infeasible using the "real" plant model. Overall, the importance of back-offs is shown to maintain feasibility given the presence of plant-model mismatch for both GP hybrid cases, however for GP non-hybrid 50 the uncertainty is too large to attain a reasonable solution.}
    \item{In Table \ref{tab:probability_OCP_time} the average computational times are between $78$ms and $174$ms. It can be seen that the GP hybrid approaches have higher computational times, which is due to the discretization required in the NMPC optimization problem. Overall the computational time of a single NMPC iteration is relatively low, while the offline computational time required to attain the back-offs is relatively high.}
\end{itemize}

\begin{figure}[ht] \centering
   \includegraphics[width=1\textwidth]{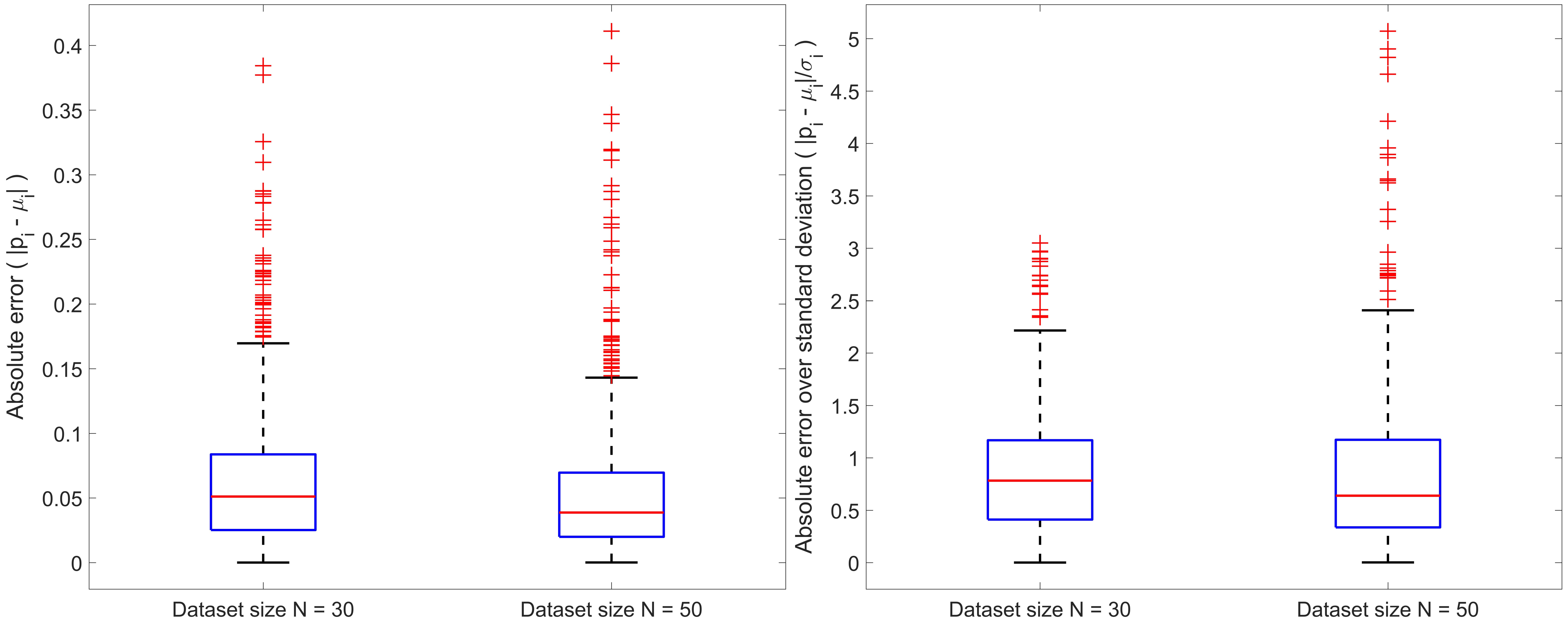}
  \caption{GP hybrid model cross-validation for dataset sizes $N=30$ and $N=50$ using 1000 randomly generated points in the same range as the training datapoints. The LHS graph shows the box plot of the absolute error, while the RHS graph shows the absolute error over the standard deviation.}
  \label{fig:Boxplot_hybrid}
\end{figure}

\begin{figure}[ht] \centering
   \includegraphics[width=1\textwidth]{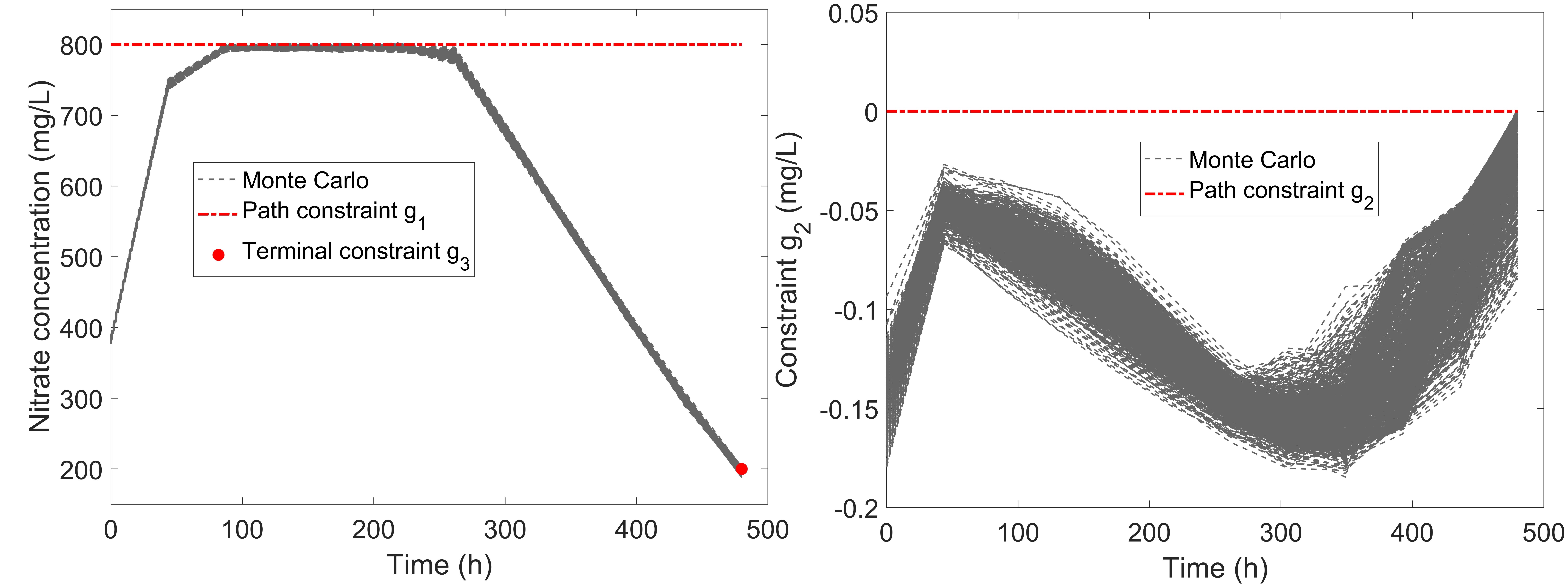}
  \caption{The 1000 MC trajectories at the final back-off iteration of the nitrate concentration for the constraints $g_1$ and $g_2$ (LHS) and the ratio of bioproduct to biomass constraint $g_2$ (RHS) for hybrid GP $N=30$.}
  \label{fig:Con_MC_30_hybrid}
\end{figure}

\begin{figure}[ht] \centering
   \includegraphics[width=1\textwidth]{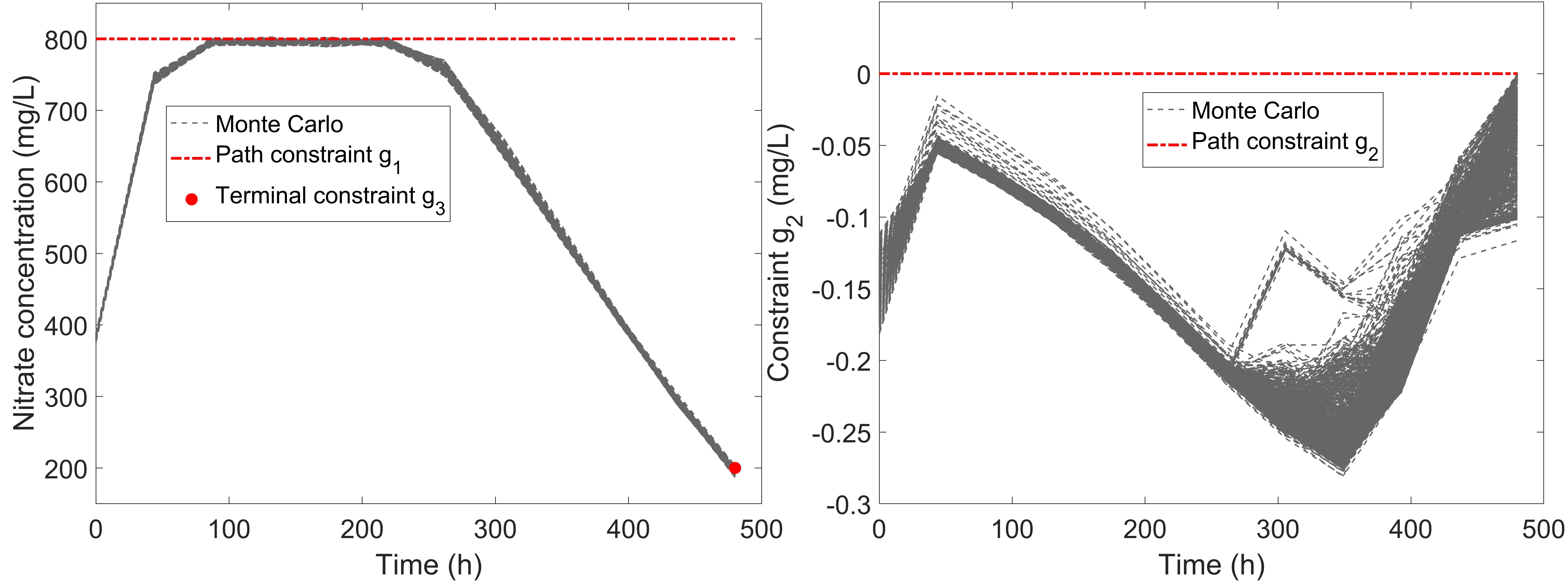}
  \caption{The 1000 MC trajectories at the final back-off iteration of the nitrate concentration for the constraints $g_1$ and $g_2$ (LHS) and the ratio of bioproduct to biomass constraint $g_2$ (RHS) for hybrid GP $N=50$.}
  \label{fig:Con_MC_50_hybrid}
\end{figure}

\begin{figure}[H] \centering
   \includegraphics[width=1\textwidth]{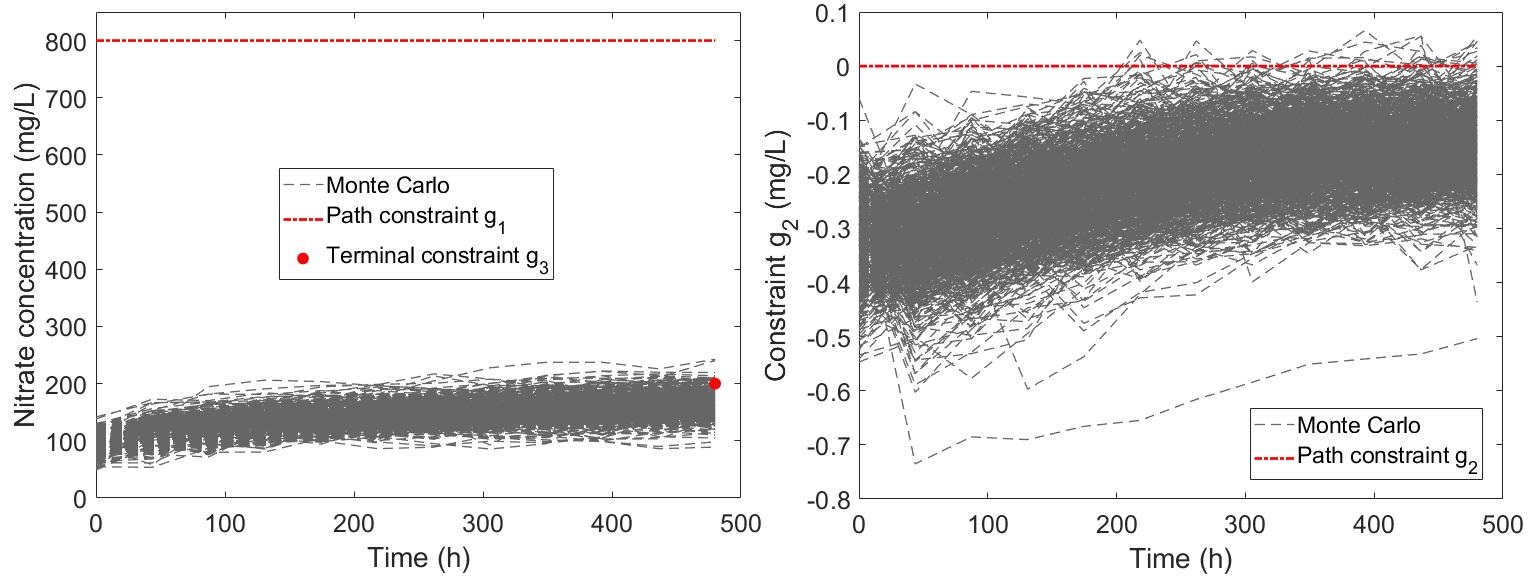}
  \caption{The 1000 MC trajectories at the final back-off iteration of the nitrate concentration for the constraints $g_1$ and $g_2$ (LHS) and the ratio of bioproduct to biomass constraint $g_2$ (RHS) for the non-hybrid GP with $N=50$ modelling the entire state space model.}
  \label{fig:Con_MC_50_nonhybrid}
\end{figure}

\begin{figure}[ht] \centering
   \includegraphics[width=1\textwidth]{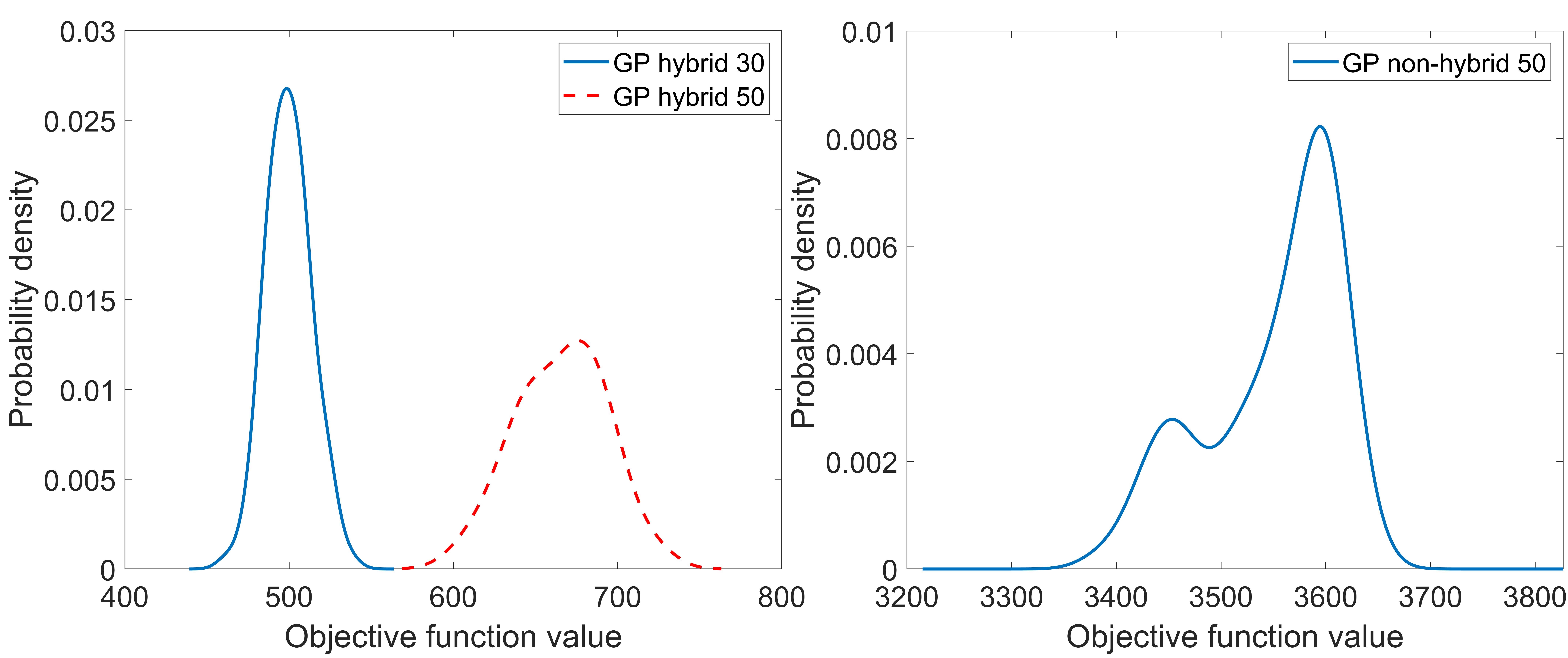}
  \caption{Probability density function for the "real" plant objective values for GP hybrid $N=30$ and $N=50$ on the LHS, and for the non-hybrid GP with $N=50$ on the RHS.}
  \label{fig:Objective_plots}
\end{figure}

\begin{figure}[ht] \centering
   \includegraphics[width=1\textwidth]{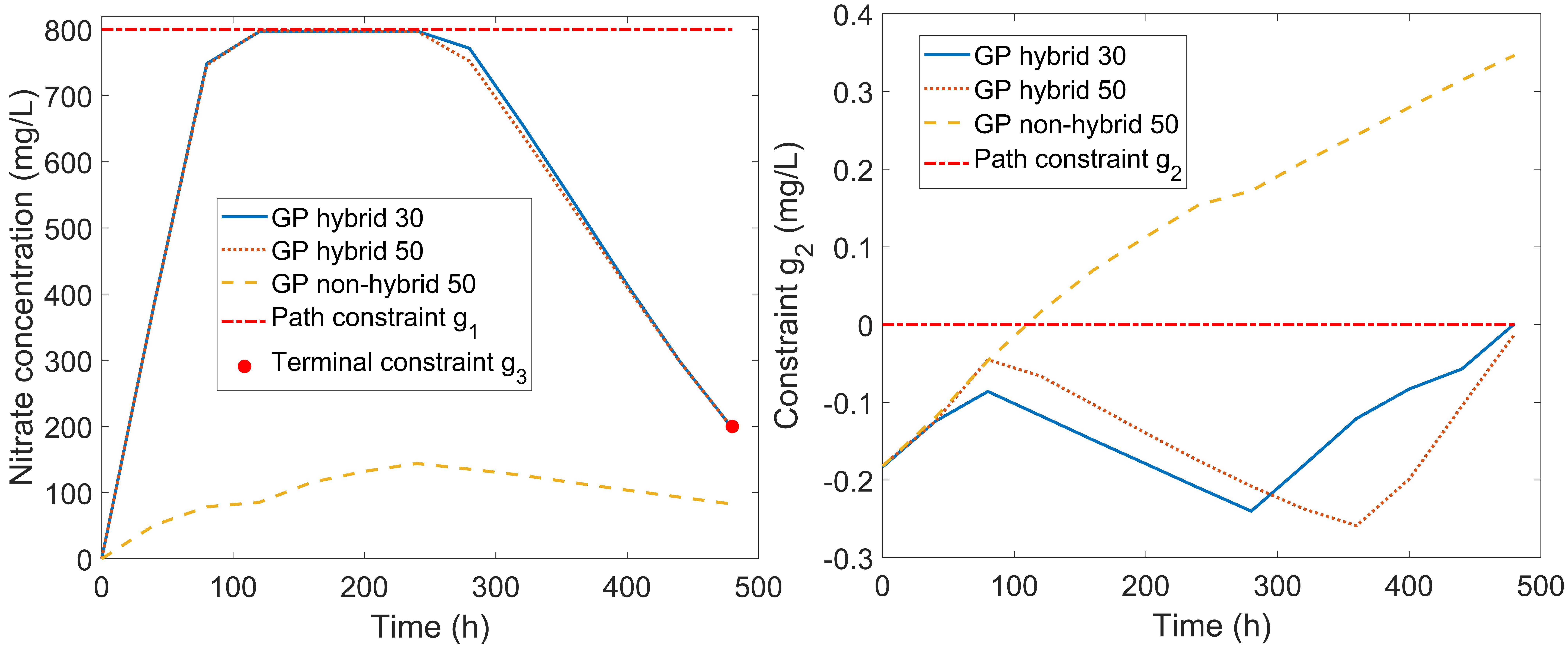}
  \caption{$90$th percentile trajectory values of the nitrate concentration for constraints $g_1$ and $g_3$ (LHS) and the ratio of the bioproduct constraint $g_2$ (RHS) for all variations applied to the "real" plant with the final tightened constraint set.}
  \label{fig:Plant_constraint_plots}
\end{figure}

\begin{figure}[ht] \centering
   \includegraphics[width=1\textwidth]{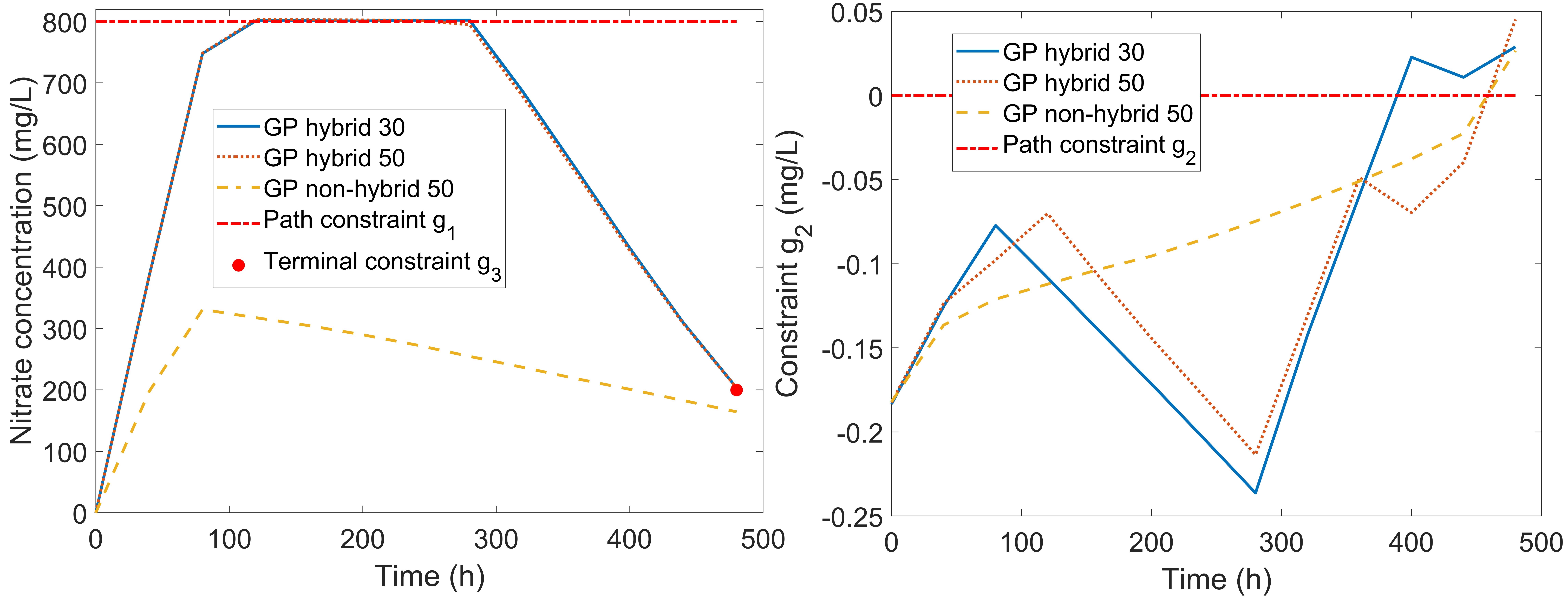}
  \caption{$90$th percentile trajectory values of the nitrate concentration for constraints $g_1$ and $g_3$ (LHS) and the ratio of the bioproduct constraint $g_2$ (RHS) for all variations applied to the "real" plant with back-off values set.}
  \label{fig:Plant_constraint_plots_nominal}
\end{figure}

\begin{table}[H]
\centering
\caption{Lower bound on the probability of satisfying the joint constraint $\hat{\beta}_{lb}$, average computational times to solve a single OCP for the GP NMPC, and the average computational time required to complete one back-off iteration.}
\begin{tabular}{*{4}{l}}
Algorithm variation & Probability $\hat{\beta}_{lb}$ & OCP time (ms) & Back-off iteration time (s) \\
\hline
GP hybrid 30        & 0.89  & 109 & 1316  \\
GP hybrid 50        & 0.91  & 174 & 2087  \\
GP non-hybrid 50    & 0.91  & 78  & 824
\end{tabular}
\label{tab:probability_OCP_time}
\end{table}

\section{Conclusions} \label{sec:conclusions}
In conclusion, a new approach is proposed to combine first principles derived models with black-box GP for NMPC. In addition, it is shown how the probabilistic nature of the GPs can be exploited to sample functions of possible dynamic models. These in turn are used to determine explicit back-offs, such that closed-loop simulations of the sampled models remain feasible to a high probability. It is shown how probabilistic guarantees can be obtained based on the number of constraint violations of the simulations. Computational time is kept low by carrying-out the constraint tightening is performed offline. Lastly, a challenging semi-batch reactor case study demonstrates the efficiency and potential for this technique to operate complex dynamic systems. 

\bibliographystyle{abbrv}
\bibliography{Backoff_GP_JPC,Hybrid_GP}
\end{document}